\documentclass{llncs}

\usepackage{txfonts}

\usepackage{stmaryrd}
\usepackage{amssymb}
\usepackage{mathtools} 
\usepackage{graphicx}
\usepackage{shuffle}
\usepackage[vlined,english]{algorithm2e}

\usepackage{paralist}
\usepackage{setspace}

\usepackage{todonotes}

\usepackage[pdftex,%
 	colorlinks,%
 	hidelinks,
 	pdfauthor={Javier Esparza, Pierre Ganty, Rupak Majumdar},%
 	hypertexnames=false,%
	pdftitle={Parameterized Verification of Asynchronous Shared-Memory Systems}]{hyperref}

\makeatletter
\def\verbatim{\small\@verbatim \frenchspacing\@vobeyspaces \@xverbatim}
\makeatother

\makeatletter
\begingroup \catcode `|=0 \catcode `[= 1
\catcode`]=2 \catcode `\{=12 \catcode `\}=12
\catcode`\\=12 |gdef|@xcomment#1\end{comment}[|end[comment]]
|endgroup
\def\@comment{\let\do\@makeother \dospecials\catcode`\^^M=10\def\par{}}
\def\begincomment{\@comment\@xcomment}
\makeatother
\newenvironment{comment}{\begincomment}{}

\newcommand{\D}{{\cal D}}
\newcommand{\US}{{\cal \D S}}
\newcommand{\T}{{\cal T}}
\newcommand{\Q}{{\cal Q}}
\newcommand{\C}{{\cal C}}

\renewcommand{\S}{{\cal S}}
\newcommand{\G}{{\cal G}}
\newcommand{\N}{{\cal N}}

\newcommand{\ES}{\ensuremath{S^E}}
\newcommand{\ESU}{\ensuremath{S^E_\D}}
\newcommand{\ESUtau}{\ensuremath{S^{\tau}_\D}}
\newcommand{\EG}{{\cal G}_E}
\newcommand{\EGU}{{\cal G}_E^\D}
\newcommand{\FW}{{\it F}}

\newcommand{\shufflep}[1]{{\shuffle}_{#1} }

\newcommand{\By}[1]{\xRightarrow{(#1)}}
\newcommand{\Bystar}[1]{\xRightarrow{(#1)}\!\!{}^*}

\newcommand{\rb}{\ensuremath{r_{\star}}}
\newcommand{\wb}{\ensuremath{w_{\star}}}

\def\prod{\mathcal{P}}

\def\set#1{{\left\{ #1 \right\}}}
\def\tuple#1{{\left\langle #1 \right\rangle}}

\def\nats{{\mathbb{N}}}

\newcommand{\fw}{{\it f}}
\newcommand{\uw}{{\it u}}

\newcommand{\proj}{\mathit{Proj}}
\newcommand{\pre}{\mathit{Pref}}

\def\by#1{{\buildrel{\smash{#1}}\over\longrightarrow}}

\def\lts#1{{\llbracket #1 \rrbracket}}

\title{\texorpdfstring{Parameterized Verification of\\ Asynchronous Shared-Memory Systems}{Parameterized Verification of Asynchronous Shared-Memory Systems}}
\author{Javier Esparza\inst{1} \and Pierre Ganty\inst{2}\thanks{Supported by the Spanish projects with references \textsc{tin}2010-20639 and \textsc{tin}2012-39391-C04.} \and Rupak Majumdar\inst{3}}
\institute{$^1$TU Munich\quad $^2$IMDEA Software Institute \quad $^3$MPI-SWS}

\begin{document}

\maketitle

\begin{abstract}
We characterize the complexity of the safety verification problem for
parameterized systems consisting of a leader process and arbitrarily many
anonymous and identical contributors.  Processes communicate through a shared,
bounded-value register.  While each operation on the register is atomic, there
is no synchronization primitive to execute a sequence of operations atomically.

We analyze the complexity of the safety verification problem when processes are
modeled by finite-state machines, pushdown machines, and Turing machines.  The
problem is coNP-complete when all processes are finite-state machines, and is
PSPACE-complete when they are pushdown machines.  The complexity remains
coNP-complete when each Turing machine is allowed boundedly many interactions
with the register.  Our proofs use combinatorial characterizations of
computations in the model, and in case of pushdown-systems, some 
language-theoretic constructions of independent interest.
%
\end{abstract}

%
%

\section{Introduction}

We conduct a systematic study of the complexity of safety verification for
\emph{parameterized asynchronous shared-memory systems}. 
These systems consist of a \emph{leader} process and
arbitrarily many identical \emph{contributors}, processes with no identity, running at
arbitrarily relative speeds and subject to faults (a process can crash). 
The shared-memory consists of a read/write register that all processes can access
to perform either a read operation or a write operation. The register is
bounded: the set of values that can be stored is finite. We do insist that
read/write operations execute atomically but sequences of operations do not: no
process can conduct an atomic sequence of reads and writes while
excluding all other processes. 
The parameterized verification problem for these systems asks to check if a 
safety property holds no matter how many contributors are present.
Our model subsumes the case in which all processes are identical by having the leader
process behave like yet another contributor. 
The presence of a distinguished leader adds (strict) generality to the problem.

We analyze the complexity of the
safety verification problem when leader and contributors are modeled by finite
state machines, pushdown machines, and even Turing machines. 
Using combinatorial properties of the model that allow simulating arbitrarily many
contributors using finitely many ones, we show that if
leader and contributors are finite-state machines the problem 
is coNP-complete.  The case in
which leader and contributors are pushdown machines
was first considered by Hague \cite{hague11}, who gave a coNP lower
bound and a 2EXPTIME upper bound. We close the gap and prove 
that the problem is PSPACE-complete.
Our upper bound requires several novel language-theoretic constructions on
bounded-index approximations of context-free languages. 
Finally, we 
address the bounded safety problem, i.e., deciding if no error can be reached by 
computations in which no contributor nor the leader 
execute more than a given number $k$ of steps
(this does not bound the length of the computation, since the number of contributors is unbounded). We show that (if $k$ is given in unary)
the problem is coNP-complete not only for pushdown machines, but also for arbitrary 
Turing machines.
Thus, the safety verification problem when the leader and contributors are poly-time Turing machines is also coNP-complete.

These results show that non-atomicity substantially reduces the complexity of
verification. In the atomic case, contributors can ensure that they are the only ones
that receive a message: the first contributor that reads the message from the store
can also erase it within the same atomic action. This allows the leader to distribute 
identities to contributors. As a consequence, 
the safety problem is at least PSPACE-hard for state machines, and undecidable for 
pushdown machines (in the atomic case, the safety problem of two pushdown machines 
is already undecidable). A similar argument shows that the bounded safety problem is
PSPACE-hard. 
In contrast, we get several coNP upper bounds, which opens the way
to the application of SAT-solving or SMT-techniques.

Besides intellectual curiosity, our work on this model is motivated by practical  
distributed protocols implemented on wireless sensor networks.
In these systems, a central co-ordinator (the base station) communicates with an arbitrary
number of mass-produced tiny agents (or motes) that run concurrently and
asynchronously.
The motes have limited computational power, and for some systems such as
vehicular networks anonymity is a requirement~\cite{VehicularLCN2007}.
Further, they are susceptible to crash faults.
Implementing atomic communication primitives in this setting is expensive 
and can be problematic: for instance, a process might crash while holding a lock.
Thus, protocols in these systems work asynchronously and without synchronization primitives.
Our algorithms provide the foundations for safety verification of these systems.


\noindent %
{\it Related Works.} %
Parameterized verification problems have been extensively studied both
theoretically and practically. It is a computationally hard problem: the
reachability problem is undecidable even if each process has a finite state
space \cite{AxKozen86}. For this reason, special cases have been extensively studied.
They vary according to the main characteristics of the systems to verify
like the communication topology of the processes (array, tree, unordered, etc); their communication primitives (shared memory, unreliable broadcasts, (lossy) queues, etc); or whether processes can distinguish from each other (using ids, a distinguished process, etc). 
Prominent examples include broadcast protocols \cite{EFM99,FL02,DimitrovaP08,DelzannoSZ10}, where
finite-state processes communicate via broadcast messages, asynchronous programs \cite{gm12,cv-tcs09},
where finite-state processes communicate via unordered channels,
finite-state processes communicating via ordered channels \cite{ACJT96},
micro architectures \cite{McMillan98},
cache coherence protocols \cite{DBLP:conf/charme/EmersonK03,Delzanno03},
communication protocols \cite{EmersonN98},
multithreaded shared-memory programs \cite{ClarkeTV08,DRV02,KKW10,TorreMP10}.

Besides the model of Hague~\cite{hague11}, the closest model to ours that has been
previously studied~\cite{GR07} is that of distributed computing with
identity-free, asynchronous processors and non-atomic registers.  The emphasis
there was the development of distributed algorithm primitives such as
time-stamping, snapshots, and consensus,
using either unbounded registers or an unbounded number of bounded registers.
\pagebreak[2]

It was left open if these
primitives can be implemented using a bounded number of bounded registers.
Our decidability results indicate that this is not possible: the safety verification problem would be undecidable if
such primitives could be implemented.

%
%

\section{Formal Model: Non-Atomic Networks}
\label{sec:prelim}

We describe our formal model, called
{\em non-atomic networks}. We take a language-theoretic view,
identifying a system with the language of its executions.

\smallskip
\noindent
\emph{Preliminaries.}
A \emph{labeled transition system} (LTS) is a quadruple
$\T = (\Sigma, Q, \delta, q_0)$, where 
 \(\Sigma\) is a finite set of \emph{action labels}, $Q$ is a (non necessarily finite) set of \emph{states},
 \(q_0 \in Q\) is the \emph{initial state},
 and $\delta \subseteq Q \times \Sigma \cup \{\varepsilon\} \times Q$ is the \emph{transition relation},
 where $\varepsilon$ is the empty string.
 We write \(q \by{a} q'\) for $(q,a,q') \in \delta$.
 For $\sigma\in\Sigma^*$, we write $q \by{\sigma} q'$ if there exist $q_1, \ldots, q_{n} \in Q$ and
 $a_0, \ldots, a_n \in \Sigma \cup \{\varepsilon\}$, $q \by{a_0} q_1 \by{a_1} q_2 \cdots q_{n} \by{a_n} q'$ such that
 $a_0 \cdots a_{n} = \sigma$.
 The sequence \(q \cdots q'\) is called a \emph{path} and \(\sigma\) its \emph{label}.
 A {\em trace} of $\T$ is a sequence $\sigma \in \Sigma^*$ such that $q_0
 \by{{\sigma}} q$ for some $q \in Q$.
 Define $L(\T)$, the {\em language} of $\T$, as the set of traces of $\T$.
 Note that \(L(\T)\) is \emph{prefix closed}: \(L(\T)=\pre(L(\T))\)
 where \(\pre(L)=\set{ s \mid \exists u\colon s\, u\in L}\)
%
%

To model concurrent executions of LTSs, we introduce two operations on languages:
the shuffle and the asynchronous product.
The {\em shuffle} of two words \(x,y\in\Sigma^*\) is the language
\(
x\shuffle y  = \{x_1y_1\dots x_ny_n\in\Sigma^* \mid \mbox{each }x_i,y_i \in\Sigma^* 
\mbox{and }x=x_1\cdots x_n\, \land y=y_1\cdots y_n\}
\).
\noindent The shuffle of two languages $L_1, L_2$
is the language \(L_1\shuffle L_2 = \textstyle{\bigcup_{x\in L_1, y\in L_2}} x \shuffle y\).
Shuffle is associative, and so we can write $L_1 \shuffle \cdots \shuffle L_n$
or $\shuffle_{i=1}^n L_i$.

The \emph{asynchronous product} of two languages \(L_1 \subseteq \Sigma_1^*\)
and \(L_2 \subseteq \Sigma_2^*\), denoted \(L_1 \parallel L_2\), is the language \(L\) over the alphabet
\(\Sigma=\Sigma_1\cup\Sigma_2\) such that \(w\in L\) if{}f
the projections of $w$ to $\Sigma_1$ and $\Sigma_2$ belong to $L_1$ and $L_2$, respectively.%
\footnote{Observe that the \(L_1\parallel L_2\) depends on \(L_1\), \(L_2\) \textbf{and} also their underlying alphabet \(\Sigma_1\) and \(\Sigma_2\).}
If a  language consists of a single word, e.g. \(L_1=\set{w_1}\), we abuse notation
and write \(w_1 \parallel L_2\).
Asynchronous product is also associative, and so we write 
$L_1 \parallel \cdots \parallel L_n$ or $\parallel_{i=1}^n L_i$.

Let $\T_1, \ldots, \T_n$ be LTSs, where $\T_i = (\Sigma_i, Q_i, \delta_i,
q_{0i})$.  The {\it interleaving} $\shuffle_{i=1}^n \T_i$ is the LTS with
actions \(\bigcup_{i=1}^n \Sigma_i\), set of states $Q_1 \times \cdots \times
Q_n$, initial state $(q_{01}, \ldots, q_{0n})$, and a transition $(q_1, \ldots,
q_n) \by{a} (q_1',\ldots, q_n')$ if{}f $(q_i, a, q_i') \in \delta_i$ for some
$1 \leq i \leq n$ and $q_j = q_j'$ for every $j \neq i$.  Interleaving models
parallel composition of LTSs that do not communicate at all.  
The language $L(\T_1 \shuffle \cdots \shuffle \T_n)$ of the interleaving is 
$\shuffle_{i=1}^n L(\T_i)$.

The {\em asynchronous parallel composition}
$\parallel_{i=1}^n \T_i$ of $\T_1, \ldots, \T_n$ is the LTS having 
\(\bigcup_{i=1}^{n} \Sigma_i\) as set of actions, 
$Q_1 \times \cdots \times Q_n$ as set of states, $(q_{01}, \ldots, q_{0n})$ as
initial state, and a transition $(q_1, \ldots, q_n) \by{a} (q_1',\ldots, q_n')$
if and only if
\begin{compactenum}
\item $a \neq \varepsilon$ and for all
$1 \leq i \leq n$ either $a \notin \Sigma_i$ and $q_i = q_i'$ or $a \in \Sigma_i$ and
$(q_i, a, q_i') \in \delta_i$, or;
\item $a = \varepsilon$, and there is $1 \leq i \leq n$ such that
$(q_i, \varepsilon, q_i') \in \delta_i$ and $q_j=q_j'$ for every $j \neq i$. 
\end{compactenum}
Asynchronous parallel composition models the parallel composition of LTSs in which 
an action $a$ must be simultaneously executed by every LTSs having $a$ in its alphabet. 
$L(\T_1 \parallel \cdots \parallel \T_n)$, the language of the asynchronous parallel composition,
is $\parallel_{i=1}^n L(\T_i)$.

\smallskip %
\noindent %
{\it Non-atomic networks.}
%
We fix a finite non-empty set $\G$ of \emph{global values}.
A {\em read-write alphabet} is any set of the form $A \times \G$, where
$A$ is a set of {\em read} and {\em write actions}, or just
{\em reads} and {\em writes}. We denote a letter $(a,g) \in A \times \G$ by $a(g)$, 
and write $\G(a_1, \ldots, a_n)$ instead of $\{a_1, \ldots, a_n\} \times \G$. 

In what follows, we consider LTSs over read-write alphabets.
We fix two LTSs $\D$ and $\C$, called the \emph{leader} and 
the \emph{contributor}, with alphabets $\G(r_d,w_d)$ and $\G(r_c, w_c)$, 
respectively, where $r_d, r_c$ are called reads and $w_c, w_d$ are called writes.  
We write \(\wb\) (respectively, \(\rb\)) to stand for either \(w_c\) or \(w_d\) 
(respectively, \(r_c\) or \(r_d\)).
We also assume that for each value \(g\in\G\)
there is a transition in the leader or contributor which reads or
writes \(g\) (if not, the value is never used and is removed from~\(\G\)).

Additionally, we fix an LTS $\S$ called a {\em  store}, whose states are the global values
of \(\G\) and whose transitions, labeled with the read-write alphabet,
represent possible changes to the global values on reads and writes.
No read is enabled initially.
Formally, the \emph{store} is an LTS
$\S = (\Sigma, \G\cup\set{g_0}, \delta_{\S}, g_0)$,
where \(\Sigma=\G(r_d, w_d, r_c, w_c)\), $g_0$ is a designated initial value not in \(\G\), and 
$\delta_{\S}$ is the set of transitions
$g \by{\rb(g)} g$ and $g' \by{\wb(g)} g$ for all $g\in \G$ and all
\(g'\in \G\cup\set{g_0}\) .
Observe that fixing \(\D\) and \(\C\) also fixes \(\S\).%
\begin{definition}\label{def:instancenetwork}
Given a pair \( (\D, \C)\) of a leader $\D$ and contributor \(\C\), 
and \(k\geq 1\), define
\(\N_k\) to be the LTS $\D \parallel \S \parallel \shufflep{k} \C $, where 
\(\shufflep{k}\C\) is \(\shuffle_{i=1}^k \C\).
The {\em (non-atomic) \((\D,\C)\)-network} \(\N\) is the set $\set{\N_k \mid k \geq 1}$,
with {\em language} $L(\N) = \bigcup_{k=1}^\infty L(\N_k)$.
We omit the prefix \((\D,\C)\) when it is clear from the context.
\end{definition}

\noindent
Notice that $L(\N_k)\; =\;  L(\D) \parallel L(\S) \parallel\,
\shufflep{k} L(\C)$ and 
$L(\N) = L(\D) \parallel L(\S) \parallel\, \shufflep{\infty} L(\C)$,
where \(\shufflep{\infty} L(\C)\) is given by \(\bigcup_{k=1}^\infty \shufflep{k} L(\C)\).

\smallskip %
\noindent %
{\it The safety verification problem.} %
A trace of a \( (\D,\C)\)-network \(\N\) is unsafe if it 
ends with an occurrence of
$w_c(\#)$, where $\#$ is a special value of \(\G\).
Intuitively, an occurrence
of $w_c(\#)$ models that the contributor raises a flag because some error has occurred.  
A \( (\D,\C)\)-network \(\N\) is \emph{safe} if{}f its language contains no
unsafe trace, namely \(L(\N)\cap \Sigma^* w_c(\#)=\emptyset\).  (We could also
require the leader to write \(\#\), or to reach a certain state; all these conditions are easily
shown equivalent.) 

Given a machine \(M\) having an LTS semantics over some read-write alphabet, we
denote its LTS by \(\lts{M}\).  Given machines $M_D$ and $M_C$ over read-write alphabets, 
The \emph{safety verification problem} for machines $M_D$ and $M_C$ consists of deciding
if the \( (\lts{M_D},\lts{M_C})\)-network is safe. Notice that the size of the input is the 
size of the machines, and not the size of the LTSs thereof, which might even be infinite.

Our goal is to characterize the complexity of the safety verification problem
considering various types of machines for the leader and the contributors. 
We first establish some fundamental combinatorial properties of non-atomic networks.

%
%

\section{Simulation and Monotonicity}
\label{sec:combi}

We prove two fundamental combinatorial properties of non-atomic networks: 
the Simulation and Monotonicity Lemmas.
Informally, the Simulation Lemma states that a leader cannot 
distinguish an unbounded number of contributors from the parallel
composition of at most $|\G|$ {\em simulators}---LTSs derivable from 
the contributors, one for each value of $\G$. 
The Monotonicity Lemma states that non-minimal traces (with
respect to a certain subword order) can be removed from a simulator without the 
leader ``noticing'', and, symmetrically, non-maximal traces can be removed 
from the leader without the simulators ``noticing''. 

\subsection{Simulation}

\noindent %
{\it First writes and useless writes.} %
Let $\sigma$ be a trace. The {\em first write} of $g$ in $\sigma$ by a contributor 
is the first occurrence of $w_c(g)$ in $\sigma$. A {\em useless write} of $g$ by a contributor
is any occurrence of $w_c(g)$ that is 
immediately overwritten by another write.
For technical reasons, we additionally assume that useless writes are not first writes.

\begin{example}
In a network trace \(w_d(g_1)_1 \, w_c(g_2)_2 \, w_c(g_3)_3\) \(r_d(g_3)_4 \, w_c(g_2)_5 \, w_c(g_1)_6\)
where we have numbered occurrences, $w_c(g_2)_2$ is a first write of $g_2$, and $w_c(g_2)_5$ 
is a useless write of $g_2$ (even though $w_c(g_2)_2$ is immediately overwritten).
\end{example}

\noindent
We make first writes and useless writes explicit by adding two new actions 
$\fw_c$ and $\uw_c$ to our LTSs, and adequately adapting the store.

\begin{definition}
The {\em extension} of an LTS $\T = (\G(r,w), Q, \delta, q_0)$ is the LTS \linebreak
$\T^E = (\G(r, w, \fw, \uw), Q, \delta^E, q_0)$, 
where $\fw, \uw$ are the {\em first write} and {\em useless write} actions, 
respectively, and 
\[
\setlength\abovedisplayskip{1pt}
\setlength\belowdisplayskip{0pt}
\delta^E = \delta \cup \{ (q, \fw(g), q'), (q, \uw(g), q') 
\mid  (q, w(g), q') \in \delta\}\enspace .\]
\end{definition}

We define an {\em extended store}, whose states are triples 
$(g, W, b)$, where $g \in \G$, 
$W \colon \G \rightarrow \{0,1\}$ is the {\em write record}, and $b \in \{0,1\}$ is the {\em useless flag}.
Intuitively, $W$ records the values written by the contributors so far. If $W(g)=0$, 
then a write to $g$ must be a first write, and otherwise a regular write or a useless write.
The useless flag is set to $1$ by a useless write, and to $0$ by other writes.
When set to $1$, the flag prevents the occurrence of a read. The flag
only ensures that between a useless write and the following write no
read happens, i.e., that a write tagged as useless will indeed be
semantically useless.  A regular or first write may be semantically
useless or not. 

\begin{definition}\label{def:extendedstore}
The {\em extended store} is the LTS
\(\ES = (\Sigma_{E}, \EG, \delta_{\ES}, c_0)\)
where
\begin{compactitem}
\item \(\Sigma_E=\G(r_d, w_d, r_c,  w_c, \fw_c, \uw_c)\);
\item $\EG$ is the set of triples $(g, W, b)$, where $g \in \G\cup\set{g_0}$, 
$W~\colon~\G~\rightarrow~\{0,1\}$, and $b \in \{0,1\}$;
\item $c_0$ is the triple $(g_0, W_0, 0)$, where $W_0(g)=0$ for every $g \in \G$;
\item $\delta_{\ES}$ has a transition $(g, W, b) \by{a} (g', W',b')$
	where \(g'\in \G\) if{}f one of the following conditions hold:
	\begin{compactitem}
	\item $a = \rb(g)$, $g'=g$, $W'=W$, and $b=b'=0$;  
	\item $a= w_d(g')$, $W' = W$ and \(b'=0\);
	\item $a= \fw_c(g')$, $W(g')=0$, $W' = W[W(g')/1]$, and \(b'=0\); 
	\item $a= w_c(g')$, $W(g')=1$, $W' = W$, and $b'=0$; 
	\item $a= \uw_c(g')$, $W(g')=1$, $W' = W$, and $b'=1$.
	\end{compactitem}
\end{compactitem}
\noindent The extension of $\N_k$
is $\N^{E}_k = \D \parallel \ES \parallel\, \shufflep{k} \C^{E}$ and
the extension of $\N$ is the set $\N^E= \{\N_k^E \mid k \geq 1 \}$. 
The languages \(L(\N^E_k)\) and \(L(\N^E)\) are defined as in Def.~\ref{def:instancenetwork}.
\end{definition}

It follows immediately from this definition that if $v \in L(\N^E)$ then the sequence
$v'$ obtained of replacing every occurrence of $\fw_c(g), \uw_c(g)$ in $v$ by $w_c(g)$ belongs to $L(\N)$. 
Conversely, every trace $v'$ of $L(\N)$ can be transformed into
a trace $v$ of $L(\N^E)$ by adequately replacing some occurrences of $w_c(g)$ by
$\fw_c(g)$ or $\uw_c(g)$. 

In the sequel, we use sequences of first writes to partition sets of  traces. 
Define \(\Upsilon\) to be the (finite) set of sequences over \(\G(\fw_c)\) with no repetitions.
By the very idea of ``first writes'' no sequence of \(\Upsilon\) writes the same value twice, hence no word in \(\Upsilon\) is longer than \(|\G|\).
Also define \(\Upsilon_{\#}\) to be those words of \(\Upsilon\) which ends with \(\fw_c(\#)\).
Given \(\tau\in\Upsilon\), define \(P_{\tau}\) to be the language given by \( (\Sigma_{E}\setminus \G(\fw_c))^{*} \shuffle \tau\).
\(P_{\tau}\) contains all the sequences over \(\Sigma_E\) in which the subsequence of first writes is exactly \(\tau\).
For \(S\subseteq \Upsilon\), \(P_{S}=\bigcup_{\sigma\in S} P_{\sigma}\).

\smallskip %
\noindent %
{\it The Copycat Lemma.} %
Intuitively, a {\em copycat} contributor is one that follows another
contributor in all its actions: it reads or writes the same
value immediately after the other reads or writes.
Informally, the copycat lemma states that any trace of a non-atomic network
can be extended with copycat contributors.

Consider first the non-extended case. Clearly, for every trace 
of $\N_k$ there is a trace of $\N_{k+1}$ in which the leader and the first $k$ 
contributors behave as before, and the $(k+1)$-th contributor behaves as a copycat of
one of the first $k$ contributors, say the $i$-th: if the $i$-th contributor executes a 
read $r_c(g)$, then the $(k+1)$-th contributor executes the same read 
{\em immediately after}, and the same for a write. 

\begin{example}
Consider the trace $r_c(g_0) \, w_d(g_1) \, r_c(g_1) \, w_c(g_2)$ of $\D \parallel \S \parallel \C$. Then the sequence $r_c(g_0)^2 \, w_d(g_1) \, r_c(g_1)^2 \, w_c(g_2)^2$ is a trace of $\D \parallel \S \parallel (\C \shuffle \C)$.
\end{example}

For the case of extended networks, a similar result holds, but the copycat
copies a first write by a regular write: 
if the $i$-th contributor executes an action other than $\fw_c(g)$,
the copycat contributor executes the same action immediately after, but 
if the $i$-th contributor executes $\fw_c(g)$, then the copycat executes $w_c(g)$. 

\begin{definition}\label{def:CompRealReachpb}
We say \(u\in \G(r_d,w_d)^*\) is {\em compatible} with a multiset $M=\set{v_1, \ldots, v_k}$ of words
over $\G(\fw_c, w_c, \uw_c, r_c)$ (possibly containing multiple copies of a word) if{}f\linebreak 
{
 $u \parallel L(\ES) \parallel \shuffle_{i=1}^k v_i \neq  \emptyset$.
Let \(\tau\in\Upsilon\). We say \(u\) is compatible with \(M\) following \(\tau\) if{}f
$P_{\tau} \cap (u \parallel L(\ES) \parallel \shuffle_{i=1}^k v_i) \neq \emptyset$.
}
\end{definition}

\begin{lemma}
	\label{lem:erasable}
Let \(u\in \G(r_d,w_d)^*\) and let \(M\) be a multiset of words
over \(\G(r_c,\fw_c,w_c,\uw_c)\).
If \(u\) is compatible with \(M\), then \(u\) is compatible with
every \(M'\) obtained by erasing symbols from \(\G(r_c)\) and
\(\G(\uw_c)\) from the words of \(M\). 
\end{lemma}
\begin{proof}
Erasing reads and useless writes (that no one reads) by contributors does not affect
the sequence of values written to the store and read by someone, hence compatibility is preserved. 
\qed
\end{proof}

\begin{lemma}[Copycat Lemma]
\label{lem:extendedmono}
Let $u \in \G(r_d,w_d)^*$, let $M$ be a multiset over $L(\C^E)$ and let \(v'\in M\). 
Given a prefix \(v\) of \(v'\) we have that
if $u$ is compatible with $M$, then $u$ is compatible with $M \oplus v[\fw_c(g)/w_c(g)]$.\footnote{Throughout the paper, we use \(\set{}\), \(\oplus\), \(\ominus\), and \(\geq\) for
the multiset constructor, union, difference and inclusion,
respectively. The word \(w[a/b]\) results from \(w\) by replacing all
occurrences of \(a\) by \(b\).}
\end{lemma}%
\begin{example}
$r_d(g_1)$ is compatible with $\fw_c(g_1) \, \fw_c(g_2)$. By the Copycat Lemma 
$r_d(g_1)$ is also compatible with $\{ \fw_c(g_1) \, \fw_c(g_2), w_c(g_1) \, w_c(g_2) \}$.
Indeed, $\fw_c(g_1) \, w_c(g_1) \, r_d(g_1) \, \fw_c(g_2) \, w_c(g_2)
\in L(\ES)$ is a trace
(even though $\fw_c(g_2)$ is useless).
\end{example}

\noindent %
{\it The Simulation Lemma.} %
The simulation lemma states that we can replace unboundedly many contributors by
a finite number of LTSs that ``simulate'' them. In particular the network is safe if{}f
its simulation is safe.

Let \(v\in L(\C^E)\).
Let $\#v$ be the number of times that actions of \(\G(\fw_c,w_c)\) occur in
$v$, minus one if the last action of $v$ belongs to \(\G(\fw_c,w_c)\). 
E.g., $\#v=1$ for $v=\fw_c(g_1) r_c(g_1)$ but $\#v=0$ for $v=r_c(g_1) \fw_c(g_1)$.
%
The next lemma is at the core of the simulation theorem.

\begin{lemma}
\label{lem:readlang}
Let \(u\in L(\D)\) and let $M=\{v_1, \ldots, v_k\}$ be a multiset over $L(\C^E)$ compatible with $u$. 
Then $u$ is compatible with a multiset \(\tilde{M}\)
over \(L(\C^E)\cap \G(r_c,\uw_c)^*\; \G(\fw_c,w_c)\).
\end{lemma}

\begin{proof}
Since $u$ is compatible with \(M\), $u \parallel L(\ES) \parallel \shuffle_{i=1}^k v_i\neq\emptyset$. 
Lemma~\ref{lem:erasable} shows that we can drop from \(M\) all the \(v_i\) such that \(v_i\in \G(r_c,\uw_c)^*\).
 Further, define $\#M = \sum_{i=1}^k \#v_i$.
We proceed by induction on $\#M$. If $\#M=0$, then all the words of $M$
belong to \(\G(r_c,\uw_c)^*\; \G(\fw_c,w_c)\), and we are done. 
If $\#M>0$, then there is \(v_i \in M\) such that $v_i = \alpha_i \,
\sigma\, \beta_i$, where $\alpha_i \in \G(r_c,\uw_c)^*$, \(\sigma\in\G(\fw_c,w_c)\), and \(\beta_i\neq \varepsilon\).  Let \(g\) be the value written by
 \(\sigma\), and let $v_{k+1}=\alpha_i w_c(g)$.
By Lemma \ref{lem:extendedmono}, $u$ is compatible
with $\{v_1, \ldots, v_{k+1}\}$, and so there is $v' \in u \parallel L(\ES)
\parallel \shuffle_{i=1}^{k+1} v_i$ in which the write \(\sigma\) of $v_i$ 
occurs in $v'$ immediately before the write of $v_{k+1}$. We now let 
the writes occur in the reverse order, which amounts to replacing 
$v_i$ by $v_i' = \alpha_i\, \uw_c(g)\, \beta_i$ and $v_{k+1}$ by 
$v_{k+1}'= \alpha_i\, \sigma$. This yields a new multiset
$M' =  M \ominus \{v_i\} \oplus \{ v'_i, v_{k+1}' \}$ compatible with $u$.
Since $\#M'=\#M-1$, we then apply the induction
hypothesis to $M'$, obtain \(\tilde{M}\) and we are done. 
\qed
\end{proof} 

\begin{definition}\label{def:simulation}
For all $g \in \G$, let \(L_g=L(\C^E)\cap \G(r_c,\uw_c)^*\, \fw_c(g)\).
Define \(S_g\) be an LTS over \(\G(r_c,\uw_c,\fw_c,w_c)\) such that \(L(S_g)=\pre(L_g\cdot w_c(g)^*)\).
Define the LTS $\N^S = \D \parallel \ES \parallel \shuffle_{g\in \G} S_g$ which we call the {\em simulation} of $\N^E$.
\end{definition}

\begin{lemma}
\label{lem:readlangbis}
Let \(u\in L(\D)\) and let $M=\{v_1, \ldots, v_k\}$ be a multiset over \(L(\C^E)\cap \G(r_c,\uw_c)^*\; \G(\fw_c,w_c)\) compatible with $u$. 
Then $u$ is compatible with a set \(S=\{s_g\}_{g\in \G}\) where \(s_g\in L(S_g)\).
\end{lemma}
\begin{proof}
	Let us partition the multiset \(M\) as \(\{M_g\}_{g\in \G}\) such that
\(M_g\) contains exactly the traces of \(M\) ending with $\fw_c(g)$ or
$w_c(g)$. Note that some \(M_g\) might be empty. 
Each non-empty \(M_g\) is of the form $M_g=\set{x_1
\fw_c(g), x_2 w_c(g), \ldots, x_n w_c(g)}$ where \(n\geq 1\), and \(x_i\in
\G(r_c,\uw_c)^*\) for every $1 \leq i \leq n$.  
Define \(M'_{g}\) as empty
if \(M_g\) is empty, and \(M'_g\) as \(M_g\) together with \(n-1\) copies of \(x_1 w_c(g)\). 
The copycat lemma shows that \(u\) is compatible with \(\oplus_{g\in\G} M'_g\).  
Let us now define the multiset \(M''_g\) to be empty if \(M'_g\) is empty, 
and the multiset of  
exactly \(n\) elements given by \(x_1 \fw_c(g)\) and \(n-1\)
copies of \(x_1 w_c(g)\) if $M'_g$ is not empty.  Again we show that \(u\) is compatible with
\(\oplus_{g\in\G} M''_g\).  The reason is that the number \(n-1\) of actions
$w_c(g)$ in each \(M''_g\) does not change (compared to \(M_g\)) and each
\(w_c(g)\) action can happen as soon as \(\fw_c(g)\) has occurred.

Now define $S$ consisting of one trace \(s_g\) for each \(g\in\G\) such
that 
\(s_g=\varepsilon\) if 
\(M''_g=\emptyset\);
and 
\(s_g=x_1\, \fw_c(g)\, w_c(g)^{n-1}\) if \(M''_g\) consists of
\(x_1\fw_c(g)\)  and \(n-1\) copies of \(x_1 w_c(g)\).

We have that \(u\) is compatible with \(S\) because the number of \(\fw_c(g)\)
and \(w_c(g)\) actions in \(M''_{g}\) and \(s_g\) does not change and each
\(w_c(g)\) action can happen as soon as \(\fw_c(g)\) has occurred. Note that
it need not be the case that \(s_g\in L(\C^E)\).
However each \(s_g \in
L(S_g)\) (recall that each \(L(S_g)\) is prefix
closed).   
\qed
\end{proof}

\begin{corollary}
\label{cor:readlang}
Let \(u\in L(\D)\) and let $M=\{v_1, \ldots, v_k\}$ be a multiset over $L(\C^E)$ compatible with $u$. 
Then $u$ is compatible with a set \(S=\{s_g\}_{g\in \G}\) where \(s_g\in L(S_{g})\).
\end{corollary}

In Lemmas~\ref{lem:erasable},\ref{lem:extendedmono},\ref{lem:readlang}
and \ref{lem:readlangbis} and Corollary~\ref{cor:readlang} compatibility is
preserved.  We can further show that it is preserved following a given sequence
of first writes. For example, in Lem.~\ref{lem:readlang} if \(u\) is compatible with \(M\) following \(\tau\) then \(u\) is compatible with \(\tilde{M}\) following \(\tau\).

\begin{lemma}[Simulation Lemma]
\label{thm:simulation_bis}
Let \(\tau\in\Upsilon\):
\[
\setlength\abovedisplayskip{1pt}
\setlength\belowdisplayskip{0pt}
L(\N^E)\cap P_{\tau}\neq\emptyset \quad\text{if{}f}\quad L(\N^S)\cap P_{\tau} \neq \emptyset\enspace .\]
\end{lemma}
\begin{proof}
($\Rightarrow$):
The hypothesis and the definition of \(\N^{E}\) shows that there is \(k\geq
1\) such that \( P_{\tau} \cap (L(\D) \parallel L(\ES) \parallel\,
\shufflep{k} L(\C^{E}))\neq \emptyset\).

Therefore we conclude that there exists \(u\in L(\D)\) and \(M=
\set{v_1,\ldots,v_{k}}\) over \(L(\C^{E})\) such that \(u\)
is compatible with \(M\) following \(\tau\).
Corollary~\ref{cor:readlang} shows that \(u\) is compatible following \(\tau\) with a set \(S=\{s_g\}_{g\in \G}\) where \(s_g\in L(S_g)\).
Therefore we have \(P_{\tau} \cap (u \parallel L(\ES) \parallel \shuffle_{g\in \G} L(S_g))\neq\emptyset\),
hence that \(P_{\tau} \cap (L(\D) \parallel L(\ES) \parallel
\shuffle_{g\in \G} L(S_g))\neq\emptyset\) and finally that \(P_{\tau}\cap L(\N^{S})\neq \emptyset\).

($\Leftarrow$):
The hypothesis and the definition of \(\N^{S}\) shows that \(P_{\tau} \cap ( L(\D) \parallel L(\ES) \parallel \shuffle_{g\in \G} L(S_g) )\neq \emptyset\).
Hence we find that there exists \(u\in L(\D)\) and a set
\(\{x_g\}_{g\in\G}\) where \(x_g\in L(S_g)\) such that \(P_{\tau}
\cap (u \parallel L(\ES) \parallel \shuffle_{g\in \G} x_g)\neq \emptyset\). The prefix closure of \(L(S_g)\) shows that either \(x_g\) does not
have a first write or $x_g = v_g \fw_c(g) w_c(g)^{n_g}$ for some \(v_g\fw_c(g) \in
L_g\) and \(n_g\in \nats\). 
In the former case, that is \(x_g\in\G(r_c,u_c)^*\),
Lemma~\ref{lem:erasable} shows that discarding the trace does not
affect compatibility.
Then define the multiset \(M\) containing for each remaining trace \(x_g=v_g \fw_c(g) w_c(g)^{n_g}\) the trace \(v_g \fw_c(g)\) and \(n_g\) traces
\(v_g w_c(g)\). \(M\) contains no other element.  Using a copycat-like argument, it is easy to show \(M\) is compatible with $u$ and further
that compatibility follows \(\tau\).
Finally, because \(v_g f_c(g) \in L(\C^{E})\cap \G(r_c,u_c)^*\,\G(f_c)\) and
because \(\C^E\) is the extension of \(\C\) we find that  every trace of \(M\)
is also a trace of \(\C^E\), hence that there exists \(k\geq 1\) such that
\(P_{\tau}\cap ( L(\D) \parallel L(\ES) \parallel\, \shufflep{k}
L(\C^E))\neq\emptyset\), and finally that \(L(\N^E)\cap
P_{\tau}\neq\emptyset\).
\qed
\end{proof}

Let us now prove an equivalent safety condition.
\begin{proposition}\label{prop:acceptsimulation}
A 
\((\D,\C)\)-network $\N$ is safe 
if{}f $L(\N^S)\cap P_{\Upsilon_{\#}}= \emptyset$.
\end{proposition}%
\begin{proof}
	From the semantics of non-atomic networks, \(\N\) is
	unsafe if and only if \(L(\N)\cap (\Sigma^* w_c(\#))\neq
	\emptyset\), equivalently, \(L(\N^{E})\cap (\Sigma_E^*
	\fw_c(\#))\neq\emptyset\) (by definition of extension),
	which in turn is equivalent to \(L(\N^{E})\cap P_{\Upsilon_{\#}}
	\neq\emptyset\) (by definition of \(P_{\Upsilon_{\#}}\)), if and
	only if \(L(\N^S)\cap P_{\Upsilon_{\#}}\neq \emptyset\) (by
	the simulation lemma). 
	\qed
\end{proof}

\subsection{Monotonicity} 

Before stating the monotonicity lemma, we need some language-theoretic definitions.
For an alphabet $\Sigma$, define the {\em subword ordering} $\mathord{\preceq} \subseteq \Sigma^*\times\Sigma^*$
on words as $u \preceq v$ if{}f \(u\) results from \(v\) by deleting some occurrences of symbols.
%
Let $L \subseteq \Sigma^*$, define $S \subseteq L$ to be
\begin{compactitem}
\item {\em cover} of $L$ if for every $u \in L$ there is $v \in S$ such that $u \preceq v$;
\item {\em support} of $L$ if for every $u \in L$ there is $v \in S$ such that $v \preceq u$.
\end{compactitem}
%
Observe that for every $u, v \in S$ such that $u \prec v$: if $S$ is a cover then so is $S \setminus \{u\}$,
and if $S$ is a support then so is $S \setminus \{v\}$. 

Recall that \(\N^S=\D\parallel \ES\parallel \shuffle_{g\in\G} S_g\).
It is convenient to introduce a fourth, redundant component that does
not change \(L(\N^S)\), but exhibits important properties of it.
Recall that the leader cannot observe the reads of the contributors,
and does not read the values introduced by useless writes. We
introduce a local copy $\ESU$ of the store with alphabet
$\G(r_d,w_d,\fw_c,w_c)$ that behaves like $\ES$ for writes and first
writes of the contributors, but has neither contributor reads nor
useless writes in its alphabet. Formally: 

\begin{definition}\label{def:leaderstore}
The {\em leader store} \(\ESU\) is the LTS $(\Sigma_D^E, \EGU, \delta_D^E, c_0)$,
\begin{compactitem}
\item \(\Sigma_D^E=\G(r_d, w_d, \fw_c, w_c)\);
\item $\EGU$ is the set of  pairs $(g, W)$, where $g \in \G\cup\set{g_0}$ and
$W~\colon~\G~\rightarrow~\{0,1\}$;
\item $c_0$ is the pair $(g_0, W_0)$, where $W_0(g)=0$ for every $g \in \G$;
\item $\delta_D^E$ has a transition $(g, W) \by{a} (g', W')$ where
	\(g'\in\G\) if{}f 
one of the following conditions hold:
\begin{inparaenum}[\itshape a\upshape)]
\item $a{=}w_d(g')$ and $W' = W$;
\item $a{=}r_d(g)$, $g'=g$, and $W'=W$;  
\item $a{=}\fw_c(g')$, $W(g')=0$, and $W' = W[W(g')/1]$; 
\item $a{=}w_c(g')$, $W(g')=1$, and $W' = W$.
\end{inparaenum}
\end{compactitem}
\end{definition}

It follows easily from this definition that $L(\ESU)$ is the projection of $L(\ES)$ 
onto \(\Sigma^E_D\), and so $L(\ES) = L(\ESU) \parallel L(\ES)$ holds.
Now, define $\US = \D \parallel \ESU$, we find that:
{
	\setlength\abovedisplayskip{1pt}
	\setlength\belowdisplayskip{0pt}
\begin{align}
 L(\N^S)
 &=  L(\D \parallel \ES \parallel \shuffle_{g\in \G} S_g) &\text{def.~\ref{def:simulation}} \notag \\
 &= L(\D) \parallel L(\ESU) \parallel  L(\ES) \parallel \shuffle_{g\in \G} L(S_g) \notag \\
 &= L(\D \parallel \ESU) \parallel  L(\ES) \parallel \shuffle_{g\in \G} L(S_g) \notag \\
 &= L(\US \parallel  \ES \parallel \shuffle_{g\in \G} S_g) \label{eq:acceptanceleaderstore}%
\end{align}%
}
\begin{lemma}[Monotonicity Lemma]\label{lem:mono}
	Let \(\tau\in \Upsilon\) and let \(\hat{L}_{\tau}\) be a cover of
	\(L(\US)\cap P_{\tau}\).  For every \(g\in \G\), let \(\underline{L}_g\) be a
	support of \( L_g\), and let \(\underline{S}_g\) be an LTS such that \(L(\underline{S}_g)=\pre(\underline{L}_g \cdot w_c^*(g))\):
{
	\setlength\abovedisplayskip{1pt}
	\setlength\belowdisplayskip{0pt}
	\begin{align*}
		(L(\US)\cap P_{\tau})          \parallel  L(\ES)  \parallel  \shuffle_{g\in\G} L(S_g)         \neq\emptyset & 
                \mbox{\; if{}f \; }  
                \hat{L}_{\tau} \parallel L(\ES) \parallel \shuffle_{g\in\G} L(\underline{S}_g) \neq\emptyset
	\enspace . \end{align*}
}
\end{lemma}

The proof of the monotonicity lemma breaks down into 
monotonicity for the contributors (Lemma~\ref{lem:contributormono}) 
and for the leader (Lemma~\ref{lem:leadermono}).

\begin{lemma}[Contributor Monotonicity Lemma]
\label{lem:contributormono}
For every \(g\in \G\), let \(\underline{L}_g\) be a
support of \( L_g\), and let \(\underline{S}_g\) be an LTS such that \(L(\underline{S}_g)=\pre(\underline{L}_g w_c^*(g))\).  Let $u \in \G(r_d, w_d)^*$ and \(\tau\in \Upsilon\):
{
	\setlength\abovedisplayskip{1pt}
	\setlength\belowdisplayskip{0pt}
\begin{align*}
	(u \parallel L(\ES) \parallel \shuffle_{g\in\G} L(S_g) )\cap P_{\tau}\neq\emptyset &
	\mbox{\; if{}f\; }  (u \parallel L(\ES) \parallel \shuffle_{g\in\G} L(\underline{S}_g) )\cap P_{\tau}\neq\emptyset \enspace .
\end{align*}
}
\end{lemma}%
\begin{proof}
	(\(\Leftarrow\)): It suffices to observe that 
	since \(\underline{L}_g\subseteq L_g\) we have 
	\(L(\underline{S}_g)\subseteq L(S_g)\) and we are done.
	(\(\Rightarrow\)): 
	Since \(L_g\subseteq \G(r_c,u_c)^* \fw_c(g)\) and \(\underline{L}_g\subseteq L_g\) we find
	that for every word \(w'\in L_g\setminus \underline{L}_g\) there
	exists a word \(w\in \underline{L}_g\) resulting from \(w'\) by erasing symbols in \(\G(u_c,r_c)\). 
	Hence, Lemma~\ref{lem:erasable} shows that erasing symbols in \(\G(u_c,r_c)\) does not affect compatibility.
	The proof concludes by observing that compatibility is further preserved for \(\tau\), and we are done.
	\qed
\end{proof}

The leader monotonicity lemma requires the following technical observation.

\begin{lemma} 
\label{lem:coversuffices}
Let \(\tau\in \Upsilon\) and \(L\subseteq \G(r_c, \fw_c, w_c, \uw_c)^*\) satisfying the following 
condition: if $\alpha \, \fw_c(g) \, \beta_1\beta_2 \in L$, then $\alpha \, \fw_c(g) \, \beta_1 \,  w_c(g)\, \beta_2 \in L$.  For every \(v, v' \in P_{\tau}\cap L(\ESU)\):
{
	\setlength\abovedisplayskip{1pt}
	\setlength\belowdisplayskip{0pt}
\begin{align*}
\mbox{if }   & v \parallel L(\ES) \parallel L \neq \emptyset & \mbox{ and } & v'\succeq v, 
&\mbox{ then }  v'\parallel L(\ES) \parallel L \neq \emptyset \enspace .
\end{align*}
}
\end{lemma}
Because \(v,v' \in P_{\tau}\cap L(\ESU)\) over alphabet
\(\Sigma^E_D=\G(r_d,w_d,\fw_c,w_c)\) and \(v'\succeq v\) we
find that \(v\) can be obtained from \(v'\) by erasing factors that are
necessarily of the form \(\wb(g)\; r_d(g)^*\) or \(r_d(g)\).
In particular \(v,v'\in P_{\tau}\) shows that
\(\proj_{\G(\fw_c)}(v)=\proj_{\G(\fw_c)}(v')=\tau\).\footnote{\(\proj_{\Sigma'}(w)\) returns the projection of \(w\) onto alphabet \(\Sigma'\).}
The proof of Lem.~\ref{lem:coversuffices} is by induction on the number of those factors. 

\begin{lemma}[Leader Monotonicity Lemma]
	\label{lem:leadermono}
	Let \(\tau\in \Upsilon\) and \(L\subseteq \G(r_c, \fw_c, w_c, \uw_c)^*\) satisfying: if $\alpha \, \fw_c(g) \, \beta_1\beta_2 \in L$, then $\alpha \, \fw_c(g) \, \beta_1 \,  w_c(g) \beta_2 \in L$.
	For every cover \(\hat{L}_{\tau}\) of  \(P_{\tau}\cap L(\US)\):
	\[
	\setlength\abovedisplayskip{1pt}
	\setlength\belowdisplayskip{0pt}
	(P_{\tau}\cap L(\US)) \parallel L(\ES) \parallel L \neq \emptyset \quad\text{if{}f}\quad \hat{L}_{\tau}\parallel L(\ES) \parallel L\neq\emptyset \enspace .\]
\end{lemma}
\begin{proof}
	(\(\Leftarrow\)): It follows from \(\hat{L}_{\tau}\subseteq (P_{\tau} \cap L(\US))\).
	(\(\Rightarrow\)): We conclude from the hypothesis that there exists \(w\in P_{\tau}\cap L(\US)\)
	such that \(w \parallel L(\ES)\parallel L\neq \emptyset\).
	Since \(\hat{L}_{\tau}\) is a cover \(P_{\tau}\cap L(\US)\), we find that there exists \(w'\in \hat{L}_{\tau}\) such that \(w'\succeq w\)
	ans \(w'\in P_{\tau}\cap L(\US)\). Finally, \(\US= \D\parallel \ESU\) shows that \(w,w'\in P_{\tau}\cap L(\ESU)\),
	hence that \(w' \parallel L(\ES) \parallel L \neq \emptyset\) following Lem.~\ref{lem:coversuffices}, and finally that \(\hat{L}_{\tau} \parallel L(\ES)\parallel L\neq\emptyset \) because \(w'\in\hat{L}_{\tau}\).
	\qed
\end{proof}

%
%

\section{Complexity of safety verification of non-atomic networks}\label{sec:complexity}

%

Recall that the \emph{safety verification problem} for machines $M_D$ and $M_C$ consists in 
deciding if the \( (\lts{M_D},\lts{M_C})\)-network is safe. 
Notice that the size of the input is the 
size of the machines, and not the size of its LTSs, which might even be infinite.
We study the complexity of safety verification for different machine classes.

Given two classes of machines \texttt{D}, \texttt{C} (like
finite-state machines or pushdown machines, see below), we define
 \emph{the class of} (\texttt{D},\texttt{C})-\emph{networks} as the set  \(
\set{(\lts{D},\lts{C})\text{-network} \mid D\in\mathtt{D},C\in\mathtt{C}}\)
and denote by \(\mathtt{Safety}(\texttt{D},\texttt{C})\) the restriction of the safety 
verification problem to pairs of machines $M_D \in \texttt{D}$ and $M_C\in \texttt{C}$.
We study the complexity of the problem when leader and
contributors are {\em finite-state machines} (FSM) and {\em pushdown machines}
(PDM).%
\footnote{We also define FSA and PDA as the automaton (i.e. language acceptor)
counterpart of FSM and PDM, respectively. As expected, definitions are
identical except for an additional accepting component given by a subset of
states in which the automaton accepts.}
In this paper a FSM is just another name for a
finite-state LTS, and the LTS \(\lts{A}\) of a FSM \(A\) is \(A\), i.e.
\(\lts{A}=A\). We define the size \(|A|\) of a FSM \(A\) as the size of
its transition relation.  
A (read/write) {\em pushdown machine} is a tuple $P = (\Q,
\G(r,w), \Gamma, \Delta, \gamma_0, q_0)$, where $\Q$ is a finite set of {\em
states} including the \emph{initial state} \(q_0\), $\Gamma$ is a {\em stack
alphabet} that contains the \emph{initial stack symbol} \(\gamma_0\), and $\Delta
\subseteq (\Q \times \Gamma) \times (\G(r,w)\cup\set{\varepsilon}) \times (\Q
\times \Gamma^*)$ is a set of {\em rules}. A {\em configuration} of a PDM $P$
is a pair $(q, y) \in \Q \times \Gamma^*$.  The LTS \(\lts{P}\) over
\(\G(r,w)\) associated to $P$ has \(\Q \times \Gamma^*\) as states,
$(q_0,\gamma_0)$ as initial state, and a transition $(q, \gamma y) \by{a}
(q',y'y)$ if{}f $(q,\gamma,a,q',y')\in \Delta$.  Define the size of a
rule \((q,\gamma,a,q',y')\in \Delta\) as \(|y'|+5\) and the size \(|P|\) of a
PDM as the sum of the size of rules in \(\Delta\).%

\smallskip %
\noindent %
{\it Determinism.} %
We show that lower bounds (hardness)
for the safety verification problems can be achieved already for {\em deterministic}
machines.
An LTS \(\T\) over a read-write alphabet is {\em deterministic} if for every 
state $s$ and every pair of transitions $s \by{a_1} s_1$ and $s \by{a_2} s_2$,
if $s_1 \neq s_2$ then $a_1$ and $a_2$ are reads, and they read different values.
Intuitively, for any state of a store \(\S\), a deterministic LTS \(\T\)
can take at most one transition in \(\S\parallel \T\).
A \( (\D,\C)\)-network is deterministic if $\D$ and \(\C\) are deterministic LTSs.
Given a class $\mathcal{X}$ of machines, we denote by $\mathrm{d}\mathcal{X}$ the
subclass of machines \(M\) of $\mathcal{X}$ such that
\(\lts{M}\) is a deterministic LTS over the read-write alphabet.
Notice that this
notion does not coincide with the usual definition of a deterministic
automaton. 

The observation is that a network with non-deterministic processes
can be simulated by deterministic ones while preserving safety; intuitively,
the inherent non-determinism of interleaving can simulate non-deterministic choice in the machines.

\begin{lemma}[Determinization Lemma]\label{lem:determinism}
There is a polynomial-time procedure that takes a pair $(\D,\C)$ of LTSs and outputs
a pair $(\D',\C')$ of deterministic LTSs such that the $(\D,\C)$-network is safe 
if{}f the $(\D',\C')$-network is safe.
\end{lemma}

We prove the lemma by eliminating non-determinism as follows.
Suppose \(\D\) is non-deterministic by having transitions \( (q,r_d(g), q')\)
and \( (q,r_d(g),q'')\).  To resolve this non-determinism, we define $\D'$ and
$\C'$ by modifying $\D$ and $\C$ as follows:
we add new states $q_1, q_2, q_3, q_4$ to $\D$ and replace the two transitions
\( (q,r_d(g),q')\) and \( (q,r_d(g),q'')\) by the transitions \(
(q,r_d(g),q_1)\), \((q_1,w_d(\mathbf{nd}),q_2)\), \((q_2,r_d(0),q_3)\),
\((q_3,w_d(g),q')\), \((q_2,r_d(1),q_4)\) and \((q_4,w_d(g),q'')\).  Let $q_0$
be the initial state of $\C$.  We add two new states $\hat{q}$ and $\tilde{q}$
to $\C$ and the transitions \( (q_0, r_c(\mathbf{nd}),\hat{q}) (\hat{q},
w_c(0), \tilde{q}) (\tilde{q},w_c(1),q_0) \).  Finally, we extend the store to
accommodate the new values $\set{0,1,\mathbf{nd}}$. It follows that \(\D'\) has
one fewer pair of non-deterministic transitions than \(\D\).
Similar transformations can eliminate other non-deterministic transitions
(e.g., two writes from a state) or non-determinism in $\C$. 

\subsection{Complexity of Safety Verification for FSMs and PDMs}

We characterize the complexity of the safety verification problem of non-atomic
networks depending on the nature of 
the leader and the contributors. We show:
\hspace*{\stretch{1}}\begin{tabular}{|lll|l|}
\hline
%
	$\mathtt{Safety}(\mathrm{dFSM},\mathrm{dFSM})$,& $\mathtt{Safety}(\mathrm{PDM},\mathrm{FSM})$,&$\mathtt{Safety}(\mathrm{FSM},\mathrm{PDM})$& $\mathrm{coNP}$-complete\\
	$\mathtt{Safety}(\mathrm{dPDM},\mathrm{dPDM})$,& $\mathtt{Safety}(\mathrm{PDM},\mathrm{PDM})$ && $\mathrm{PSPACE}$-complete\\
\hline
\end{tabular}\hfill\vspace{0pt}

\begin{theorem}{}
\label{th:nphard-dfa}
\(\mathtt{Safety}(\mathrm{dFSM},\mathrm{dFSM})\) is $\mathrm{coNP}$-hard.
\end{theorem}
We show hardness by a reduction from 3SAT to the complement of the safety verification
problem.
Given a 3SAT formula, we design a non-atomic network in which the leader and contributors
first execute a protocol that determines an assignment to all variables, and uses subsets
of contributors to store this assignment.
For a variable $x$, the leader writes $x$ to the store, inviting proposals for values.
On reading $x$, contributors non-deterministically write either ``$x$ is 0'' or ``$x$ is 1''
on the store, possibly over-writing each other.
At a future point, the leader reads the store, reading the proposal that was last written,
say ``$x$ is 0.''
The leader then writes ``commit $x$ is 0'' on the store.
Every contributor that reads this commitment moves to a state where it returns $0$ every time
the value of $x$ is asked for.
Contributors that do not read this message are stuck and do not participate further.
The commitment to $1$ is similar.
This protocol ensures that each variable gets assigned a consistent value. 

Then, the leader checks that each clause is satisfied by 
querying the contributors for the values of variables (recall that contributors
return consistent values) and checking each clause locally. 
If all clauses are satisfied, the leader writes a special symbol \(\#\).
The safety verification problem checks that \(\#\) is never written, which happens
if{}f the formula is unsatisfiable.
Finally, Lemma~\ref{lem:determinism} ensures all processes are deterministic.

\begin{theorem}{}
\label{thm:npeasy}
$\mathtt{Safety}(\mathrm{PDM},\mathrm{FSM})$ is in $\mathrm{coNP}$.
\end{theorem}
\begin{proof}
Fix a \((\D,\C)\)-network \(\N\), where 
$P_D$ is a PDM generating $\D=\lts{P_D}$, and \(\C\) is a FSM.
Hence $L(\D)$ is a context-free language and $L(\C)$ is regular.
Prop.~\ref{prop:acceptsimulation} and
Def.~\ref{def:simulation} (of \(\N^S\)) show
that the \((\D,\C)\)-network \(\N\) is accepting if{}f 
\(L(\D\parallel \ES \parallel \shuffle_{g\in \G} S_g)\cap P_{\Upsilon_{\#}}\neq \emptyset\).
Since \(\C\) is given by a FSM, so is \(\C^{E}\).
Further,
\(L_g=L(\C^E)\cap \G(r_d,u_c)^* \fw_c(g)\) has a support captured by those paths in \(\C^E\) starting from the initial state and whose label ends by
\(\fw_c(g)\) and in which no state is entered more than once.
	Therefore if \(\C^E\) has \(k\) states then the set of paths starting from the initial state, of length at most \(k+1\)
	and whose label ends with \(\fw_c(g)\) is a support, call it \(\underline{L}_g\), of  \(L_g\). 
	Next, Lem.~\ref{lem:contributormono} shows that deciding \(L(\D\parallel \ES
\parallel \shuffle_{g\in \G} S_g)\cap P_{\Upsilon_{\#}}\neq
\emptyset\) is
	equivalent to \(L(\D\parallel \ES \parallel \shuffle_{g\in \G}
	\pre(\underline{L}_g\cdot w_c(g)^*))\cap P_{\Upsilon_{\#}}\neq \emptyset\).

	Note that this last check does not directly provide a NP algorithm for non-safety because,
	due to the write records, \(\ES\) is exponentially larger than \(|\G|\).
	So, we proceed by pushing down sequences of first writes and obtain the following equivalent statement: 
	\(L(\D)\parallel (L(\ES)\cap P_{\Upsilon_{\#}}) \parallel (\shuffle_{g\in \G} L(\pre(\underline{L}_g\cdot w_c(g)^*)) \cap P_{\Upsilon_{\#}}) \neq \emptyset\).
	
	Now, we get an NP algorithm as follows: 
\begin{inparaenum}[\upshape(\itshape a\upshape)]
\item 	guess \(\tau\in \Upsilon_{\#}\) (this can be done in time polynomial in \(|\G|\));
\item 	construct in polynomial time a FSA \(A_1\) for \(L(\ES)\cap P_{\tau}\)
	(\(A_1\) results from \(\ES\) by keeping the \(|\tau|\) write records corresponding to \(\tau\));
\item   for each \(g\in \tau\), guess \(z_g \in \underline{L}_g\) 
	(the guess can be done in polynomial time);
\item   guess \(z\in (\shuffle_{g\in \G} z_g)\cap P_{\tau}\) 
	(this fixes a sequence of reads, useless writes and first writes of the contributors according to \(\tau\));
\item 	construct in polynomial time a FSA \(A_2\) such that \(L(A_2)\) is the
	least language containing \(z\) and if \(\alpha \fw_c(g) \beta_1 \beta_2 \in L(A_2)\) then \(\alpha \fw_c(g) \beta_1 w_c(g) \beta_2 \in L(A_2)\)
	(intuitively we add selfloops with write actions of \(\G(w_c)\) to the FSA accepting \(z\) such that \(w_c(g)\) occurs provided \(\fw_c(g)\) has previously occurred);
\item   construct in time polynomial in \(|P_D|\) a context-free grammar (CFG) \(G_D\) such that \(L(G_D)=L(P_D)\);
\item 	construct in polynomial time a CFG \(G\)  such that \(L(G)=L(G_D)\parallel L(A_1) \parallel L(A_2)\) 
	(this can be done in time polynomial in \(|G_D|+|A_1|+|A_2|\) as stated in Prop.~\ref{prop:bowtie}, Sect.~\ref{sec:language});
\item 	check in polynomial time whether \(L(G)\neq \emptyset\).
\end{inparaenum} 
\qed
\end{proof}

We continue with the following results showing that even if all processes but
the leader are given a stack then the safety verification problem remains in
coNP. A detailed proof is given in Appendix~\ref{sec:npNotSoEasiness}.

\begin{theorem}{}
\label{thm:npnotsoeasy}
$\mathtt{Safety}(\mathrm{FSM},\mathrm{PDM})$ is in $\mathrm{coNP}$.
\end{theorem}

The complexity of the problem becomes higher when all the processes are PDMs.

\begin{theorem}{}
\label{th:pspace-hard}
$\mathtt{Safety}(\mathrm{dPDM}, \mathrm{dPDM})$ is $\mathrm{PSPACE}$-hard.
\end{theorem}
PSPACE-hardness is shown by reduction from the acceptance problem of a polynomial-space
deterministic Turing machine.
The proof is technical.
The leader and contributors simulate steps of the Turing machine in rounds.
The stack is used to store configurations of the Turing machine.
In each round, the leader sends the current configuration of the Turing machine
to contributors by writing the configuration one element at a time on to the store
and waiting for an acknowledgement from some contributor that the element was received.
The contributors receive the current configuration and store the next configuration
on their stacks.
In the second part of the round, the contributors send back the configuration to the leader. 
The leader and contributors use their finite state to make sure all elements of the configuration are sent
and received.

Additionally, the leader and the contributors use the stack to count to $2^n$ steps.
If both the leader and some contributor count to $2^n$ in a computation, 
the construction ensures that the Turing machine has been correctly simulated for
$2^n$ steps, and the simulation is successful.
The counting eliminates bad computation sequences in which 
contributors conflate the configurations from different steps due to asynchronous reads
and writes.

Next we sketch the upper PSPACE bound that uses constructions on approximations
of context-free languages. The details of the proof are available in
Appendix~\ref{sec:language}.



%

\begin{theorem}{}
\label{thm:pspaceeasy}
$\mathtt{Safety}(\mathrm{PDM},\mathrm{PDM})$ is in $\mathrm{PSPACE}$.
\end{theorem}%
\begin{proof}
Let $P_D$ and $P_C$ be PDMs respectively generating $\D=\lts{P_D}$
and $\C=\lts{P_C}$, hence $L(\D)$ and $L(\C)$ are context-free languages.
Proposition~\ref{prop:acceptsimulation} shows
that the \((\D,\C)\)-network \(\N\) is accepting if{}f \(L(\N^S)\cap P_{\Upsilon_{\#}}\neq
\emptyset\) if{}f \(L(\US\parallel \ES \parallel \shuffle_{g\in \G}
S_g)\cap P_{\Upsilon_{\#}}\neq \emptyset\) (by
\eqref{eq:acceptanceleaderstore}).
From the construction of the Simulation
Lemma, for each $g \in \G$ the language $L_g=L(\C^E)\cap \G(r_d,u_c)^*
\fw_c(g)$ is context-free, and so is \(L(S_g)\). Given \(P_C\) we compute in polynomial time a PDA $P_g$ such
that \(L(P_g)=L_g\). Next, 
{
\setlength\abovedisplayskip{1pt}
\setlength\belowdisplayskip{0pt}
\begin{align}
	& L(\US\parallel \ES \parallel \shuffle_{g\in \G} S_g)\cap P_{\Upsilon_{\#}}\neq \emptyset \notag\\
	\text{if{}f }& (L(\US) \cap P_{\Upsilon_{\#}}) \parallel
	L(\ES) \parallel \shuffle_{g\in \G} L(S_g) \neq \emptyset \notag\\
	\text{if{}f }& (L(\US) \cap P_{\Upsilon_{\#}}) \parallel
	L(\ES) \parallel \shuffle_{g\in \G} \pre(L(P_g)\cdot
	w_c(g)^*) \neq \emptyset \label{eq:pspace-step3}\\
	\text{if{}f }& (\textstyle{\bigcup_{\tau\in\Upsilon_{\#}}}
	\hat{L}_{\tau}) \parallel L(\ES) \parallel\shuffle_{g\in \G} \pre(\underline{L(P_g)}\cdot w_c(g)^*) \neq \emptyset \label{eq:continuebis}
\end{align}
}
\eqref{eq:pspace-step3} follows from definition of \(S_g\) and \(L_g=L(P_g)\);
\eqref{eq:continuebis}  follows from Lem.~\ref{lem:mono} and by letting \(\hat{L}_\tau\) and \(\underline{L(P_g)}\) be a cover and support of \(L(\US)\cap P_{\tau}\) and \(L(P_g)\), respectively.

Next, for all \(g\in \G\) we compute a FSA \(A_g\) such that \(L(A_g)\) is a
support of \(L(P_g)\).
    Our first language-theoretic construction shows 
    that the FSA \(A_g\) can be computed in time exponential but space polynomial in \(|P_g|\).
Then, because $L(\ES)$ is a
regular language, we compute in space polynomial in
\(|P_D|+|P_C|\) a FSA $A_C$ such that $L(A_C) = L(\ES) \parallel \shuffle_{g\in \G}
\pre(L(A_g)\cdot w_c(g)^*)$.  
Hence, by \eqref{eq:continuebis} and because
of \(\Upsilon_{\#}\) (guessing and checking \(\tau\in\Upsilon_{\#}\) is done in time polynomial in \(|\G|\)) we find that it
suffices to prove \(\hat{L}_{\tau} \parallel L(A_C)\neq\emptyset\) is decidable
in space polynomial in \(|P_D|+|P_C|\).

To compute a cover \(\hat{L}_{\tau}\) of \(L(\US)\cap P_{\tau}\), 
we need results about the $k$-index approximations of a context-free language
\cite{Brainerd1967}. 
Given a CFG $G$ in CNF and $k \geq 1$, we define the {\em $k$-index approximation} of $L(G)$, denoted by $L^{(k)}(G)$, 
consisting of the words of $L(G)$ having a derivation in which every
intermediate word contains at most $k$ occurrences of variables. 
We further introduce an operator $\bowtie$ which, given $G$ and FSA $A$, computes in polynomial time
a context-free grammar $G \bowtie A$ such that $L(G \bowtie A) = L(G) \parallel L(A)$.
We prove the following properties:
\begin{enumerate}
\item \label{item:one} $L^{(3m)}(G)$ is a cover of $L(G)$, where $m$ is the number of variables of $G$; 
\item \label{item:two} for every FSA $A$ and \(k\geq 1\), $L^{(k)}(G \bowtie A) = L^{(k)}(G) \parallel L(A)$;
\item \label{item:three} $L^{(k)}(G)\neq\emptyset$ on input $G, k$ can be decided in \(\mathrm{NSPACE}(k \log(|G|))\).
\end{enumerate}

Equipped with these results, the proof proceeds as follows. Let $G_D$ be a
context-free grammar such that $L(G_D) = L(P_D)$. 
It is well-known that $G_D$ can be computed in time polynomial in $|P_D|$. 
Next, given \(\tau\), we compute a grammar $G_D^\tau$ recognizing $P_{\tau} \cap L(\US)$ as follows. 
The definition of \(\US\) shows that \(P_{\tau} \cap L(\US)=L(\D)\parallel (L(\ESU)\cap P_{\tau})\).
We then compute a FSA \(\ESUtau\) such that \(L(\ESUtau)=L(\ESU)\cap P_{\tau}\). 
It can be done in time polynomial in \(|P_D|+|P_C|\) because it is a restriction of \(\ES\) where write records are totally ordered
according to \(\tau\) and there are exactly \(|\tau|\) of them. 
Therefore we obtain, 
\(P_{\tau} \cap L(\US)=L(G_D)\parallel L(\ESUtau)\) because \(L(G_D)=L(\D)\). 
Define \(G^{\tau}_D\) as the CFG \(G_D\bowtie \ESUtau\) which can be computed in polynomial time
in \(G_D\) and \(\ESUtau\), hence in \(|P_D|+|P_C|\). 
Clearly \(L(G_D^{\tau})= P_{\tau} \cap L(\US)\).
Further, $L^{(k)}(G_D^\tau)$ is a cover of $L(G_D^\tau)$ for some $k \leq p(|P_D|)$,
where $p$ is a suitable polynomial.

By item \ref{item:two}, $L^{(k)}(G_D^\tau)
\parallel L(A_C) = L^{(k)}( G_D^\tau \bowtie A_C)$, where the grammar $G_D^\tau
\bowtie A_C$  can be constructed in exponential time and space polynomial in $|P_D|+|P_C|$. 
Now we apply a generic
result of complexity (see e.g. Lemma 4.17, \cite{AroraB09}), slightly adapted: given functions $f_1, f_2 \colon \Sigma^* \rightarrow \Sigma^*$ and $g \colon \Sigma^* \times \Sigma^*
\rightarrow \Sigma^*$  if $f_i$ can be computed by a 
 $s_{f_i}$-space-bounded Turing machine, and $g$ can be computed by a
$s_{g_1}(|x_1|)\cdot s_{g_2}(|x_2|)$-space-bounded Turing machine, then 
$g(f_1(x),f_2(x))$ can be computed in  $\log(|f_1(x)|+|f_2(x)|) + s_{f_1}(|x|) + s_{f_2}(|x|) + s_{g_1}(|f_1(x)|) \cdot s_{g_2}(|f_2(x)|))$ space. 
We have
\begin{compactitem}
\item $f_1$ is the function that computes $G_D^\tau \bowtie A_C$ 
	on input $(P_D, P_C)$, and $f_2$ is the function that on input \(P_D\) computes $3m$, where 
$m$ is the number of variables of \(G_D^\tau\). 
So the output size of $f_1$ is exponential in the input size, while it is polynomial for \(f_2\).
Moreover, \(s_{f_i}\) for \(i=1,2\) is polynomial.
\item $g$ is the function that on input $(G_D^\tau \bowtie A_C, 3m)$ 
yields $1$ if $L^{(3m)}( G_D^\tau \bowtie A_C) \neq \emptyset$, and $0$ otherwise,
where $m$ is the number of variables of $G_D^\tau$. By (\ref{item:three})
$s_{g_1}$ is logarithmic, and $s_{g_2}$ is linear. 
\end{compactitem}
\noindent Finally, the generic complexity result shows that $g \circ f$ can be computed in space polynomial in \(|P_D|+|P_C|\),
and we are done.
\qed
\end{proof}

We note that our three language-theoretic constructions 
(the construction of automaton $A_g$ that is a cover of $L(P_g)$ of size at most exponential
in $|P_g|$,
and results~\ref{item:one}, \ref{item:two}, and~\ref{item:three} in the proof above)
improve upon previous constructions, and are all required for the optimal upper bound.
Hague \cite{hague11} shows an alternate doubly exponential construction
using a result of Ehrenfeucht and Rozenberg in place of 
Theorem~\ref{thm:min-cfl}. 
This gave a 2EXPTIME algorithm.  
Even after using our exponential time construction for $A_g$, 
we can only get an EXPTIME algorithm,
since the non-emptiness problem for (general) context-free languages is P-complete \cite{Jonespcomplete}. 
Our bounded-index approximation for the cover and the space-efficient
emptiness algorithm for bounded-index languages are crucial to the PSPACE upper bound.

\subsection{The bounded safety problem}

Given $k > 0$, we say that a \( (\D,\C)\)-network is \emph{$k$-safe} if all traces in which
the leader and each contributor make at most $k$ steps
are safe; i.e., we put a bound of
$k$ steps on the runtime of each contributor, and consider only safety within this bound. 
Here, a step consists of a read or a write of the shared register.
The bound does not limit the total length of traces, because the number of contributors is unbounded. 
The {\em bounded safety} problem asks, given $\D$, $\C$, and $k$ written in unary, 
if the \( (\D,\C)\)-network is $k$-safe.

Given a class of (\texttt{D},\texttt{C})-networks, we 
define \(\mathtt{BoundedSafety}(\texttt{D},\texttt{C})\) as the restriction of the $k$-safety 
problem to pairs of machines $M_D \in \texttt{D}$ and $M_C\in \texttt{C}$, where we write $k$ in unary. 
A closer look to Theorem \ref{th:nphard-dfa} shows that its proof reduces the satisfiability problem for 
a formula $\phi$ to the bounded safety problem for a (\texttt{D},\texttt{C})-network and a number
$k$, all of which have polynomial size $|\phi|$. This proves that 
\(\mathtt{BoundedSafety}(\mathrm{dFSM},\mathrm{dFSM})\) is coNP-hard. 
We show that, surprisingly, bounded safety remains coNP-complete for pushdown systems, and, even further, for
arbitrary Turing machines. 
Notice that the problem is already coNP-complete for one single Turing machine.

We sketch the definition of the Turing machine model (TM), which
differs slightly from the usual one. Our Turing machines have two kind
of transitions: the usual transitions that read and modify the
contents of the work tape, and additional transitions with labels
in $\G(r,w)$ for communication with the store. The machines are input-free, 
i.e., the input tape is always initially empty.

\begin{theorem}{}
	$\mathtt{BoundedSafety}(\mathrm{TM},\mathrm{TM})$ is $\mathrm{coNP}$-complete.
\end{theorem}
\begin{proof}
Co-NP-hardness follows from Theorem \ref{th:nphard-dfa}.
To prove $\mathtt{BoundedSafety}(\mathrm{TM},\mathrm{TM})$ is in NP we use the simulation lemma. Let
$M_D, M_C, k$ be an instance of the problem, where $M_D, M_C$ are Turing machines
 of sizes $n_D, n_C$ with LTSs $\D=\lts{M_D}$ and $\C=\lts{M_C}$, and let $n_D+n_C=n$. 
In particular, we can assume $|\G| \leq n$, because we only need to
consider actions that appear in $M_D$ and $M_C$.
If the \( (\D,\C)\)-network is not $k$-safe, then by definition there exist
\(u\in L(\D)\) and a multiset $M=\{v_1, \ldots, v_k\}$ over $L(\C^E)$ 
such that \(u\) is compatible with \(M\) following some \(\tau\in\Upsilon_{\#}\);
moreover, all of $u,v_1, \ldots, v_m$ have length at most $k$. 
By Cor.~\ref{cor:readlang} and Lem.~\ref{lem:erasable} (showing we can drop traces
without a first or regular write), there exists a set \(S= \{s_{g_1}, \ldots,
s_{g_m} \}\) with \(m\leq |\G|\leq n\), where $s_{g_i} \in L_{g_i}\cdot w_c(g_i)^*$, 
and numbers $i_1, \ldots, i_m$ such that $u$ is compatible with 
$\{s_{g_1}\, w_c(g_1)^{i_1}, \ldots, s_{g_m}\, w_c(g_m)^{i_m} \}$ following \(\tau\).
Since each of the $s_{g_i}$ is obtained by suitably renaming the actions of a trace,
we have $|s_{g_i}| \leq k$. Moreover, since the $w_c(g_j)^{i_j}$ parts provide the writes
necessary to execute the reads of the $s_g$ sequences, and there are at most \(k
\cdot (m+1) \leq k \cdot (n+1) \) of them, the numbers can be chosen so
that $i_1, \ldots, i_m \leq O(n \cdot k)$ holds. 

We present a nondeterministic polynomial algorithm that decides if the \( (\D,\C)\)-network is
$k$-unsafe. The algorithm guesses \(\tau\in\Upsilon_{\#}\) and traces
$u, s_{g_1}, \ldots, s_{g_m}$ of length at most $k$.  Since there are at most $n+1$ of those traces,
this can be done in polynomial time. Then, the algorithm guesses numbers $i_1, \ldots, i_m$.  
Since the numbers can be chosen so that $i_1, \ldots, i_m \leq O(n \cdot k)$, this can also
be done in polynomial time. Finally, the
algorithm guesses an interleaving of
\(u, s_{g_1}\, w_c(g_1)^{i_1}, \ldots, s_{g_m}\, w_c(g_m)^{i_m}\) and
checks compatibility following \(\tau\). This can be done in
$O(n^2 \!\cdot\! k)$ time. If the algorithm succeeds, then there is a witness that $(L(\D) \parallel
L(\ES) \parallel \shuffle_{g\in\G} L(S_g))\cap P_{\tau} \neq \emptyset$ holds,
which shows, by Prop.~\ref{prop:acceptsimulation}
and Def.~\ref{def:simulation} (of \(\N^S\)) that the \( (\D,\C)\)-network is unsafe.
\qed
\end{proof}

A TM is poly-time if it
takes at most $p(n)$ steps for some polynomial $p$, where
$n$ is the size of (the description of) the
machine in some encoding. 
As a corollary, we get that the safety verification problem when
leaders and contributors are poly-time Turing machines is $\mathrm{coNP}$-complete.
Note that the coNP upper bound holds even though the LTS corresponding to a poly-time TM is
exponentially larger than its encoding.


{\small
\bibliographystyle{splncs03}
\bibliography{ref}
}


\appendix 

\section{Combinatorics}

\begin{proof}[of Lem~\ref{lem:coversuffices}]
	Because \(v,v' \in P_{\tau}\cap L(\ESU)\) over alphabet
	\(\Sigma^E_D=\G(r_d,w_d,\fw_c,w_c)\) and \(v'\succeq v\) we
	find that \(v\) can be obtained from \(v'\) by erasing factors that are
	necessarily of the form \(\wb(g)\; r_d(g)^*\) or \(r_d(g)\).
	In particular \(v,v'\in P_{\tau}\) shows that
	\(\proj_{\G(\fw_c)}(v)=\proj_{\G(\fw_c)}(v')=\tau\).\footnote{\(\proj_{\Sigma'}(w)\) returns the projection of \(w\) onto alphabet \(\Sigma'\).}

	The proof is by induction on the number \(m\) of those factors.
	If \(m=0\) then \(v'=v\) and we are done.  Now, let \(m>0\) and let
	\(v_{\dag}\) be the trace which results from erasing one factor
	\(\sigma \in (\G(w_d,w_c) \G(r_d)^*) \cup (\G(r_d)) \) from \(v'\).
	That is \(v'=v_{\dag}^1 \sigma v_{\dag}^2\) where \(v_{\dag}^1
	v_{\dag}^2=v_{\dag}\succeq v\).
	Without loss of generality we can assume that if \(\sigma\notin \G(r_d)\) then the
	first symbol of \(v_{\dag}^2\) is a write action (it belongs to \(\G(\fw_c,w_c,w_d)\)).
	Also observe that by \(\ESU\), if \(\sigma\in \G(r_d)\) then the last write in \(v_{\dag}^1\) writes the value needed by \(\sigma\).

	Since \(v'\in P_{\tau}\cap L(\ESU)\), it is routine to check that \(v_{\dag}\in P_{\tau}\cap L(\ESU)\).
	Therefore we conclude from the induction hypothesis and \( v_{\dag}\succeq v\) that there exists
	a trace \(w_{\dag}\in L(\ES) \parallel L\) such that
	\(v_{\dag}\parallel w_{\dag}\neq \emptyset\), equivalently that 
	\(\proj_{\Sigma_D^E}(w_{\dag})=v_{\dag}\) since
	\(\Sigma^E_D\subseteq \Sigma_E\). Observe that \(w_{\dag}\) can be divided into \(w_{\dag}^1 w_{\dag}^2\) such that
	\(\proj_{\Sigma_D^E}(w_{\dag}^i)=v_{\dag}^i\) for \(i=1,2\). 
	Let \(w_{\ddag}^1\) be the (possibly empty) suffix of \(w_{\dag}^1\) starting
	at the last occurrence of an action of \(\Sigma_D^E\) and let
	\(w_{\ddag}^2\) be the prefix of \(w_{\dag}^2\) which ends at the first
	occurrence of an action of \(\Sigma_D^E\).  Then \(L(\ES)\) shows that
	\(w_{\ddag}^1 w_{\ddag}^2\) belongs to \(\G(r_c)^* \G(\uw_c)^*\).

	Let us now consider the added factor \(\sigma\).  First, let us notice that
	\(\proj_{\Sigma_D^E}(w_{\dag}^1 \sigma w_{\dag}^2)= v_{\dag}^1 \sigma
	v_{\dag}^2=v'\). Thus it suffices to show that  \(w_{\dag}^1 \sigma
	w_{\dag}^2\in L(\ES)\parallel L\).
	
	Now, if \(\sigma\in \G(r_d)\), we can choose \(w_{\dag}^1\) and \(w_{\dag}^2\) such that \(w_{\ddag}^1=\varepsilon\). 
        Notice that as for \(v_{\dag}^1\) and \(v_{\dag}^2\),  then the last write that occurs in \(w_{\dag}^{1}\) 
	writes the value needed by \(\sigma\) and so \(w_{\dag}^1 \sigma
	w_{\dag}^2\in L(\ES)\).
	Hence we find that \(w_{\dag}^1 \sigma w_{\dag}^2 \in L(\ES)\parallel L\)
	because the alphabet of \(L\) is disjoint from \(\G(r_d)\).

	On the other hand if \(\sigma\notin \G(r_d)\) we can choose \(w_{\dag}^1\) and \(w_{\dag}^2\) 
	such that \(w_{\ddag}^2=\varepsilon\). Then \(\sigma=\wb(g)r_d(g)^i\) for
	some \(i\in\nats\).  Also \(w_{\dag}^1 \sigma w_{\dag}^2\in L(\ES)\) because, as assumed above, the first symbol of \(w_{\dag}^2\) is a write action.

	Finally, if \(\sigma=w_d(g)r_d(g)^i\) we find that \(w_{\dag}^1 \sigma w_{\dag}^2 \in L(\ES) \parallel L\) because
	the alphabet of \(L\) is disjoint from \(\G(w_d,r_d)\).
	Else if \(\sigma=w_c(g)r_d(g)^i\) we find, by \(\ES\), that \(\fw_c(g)\) must occur in \(w_{\dag}^1\), hence
	that \(w_c(g)\) can be matched in \(L\) following the hypothesis on \(L\). 
	\qed
\end{proof}

%
%
\section{coNP lower bound}

\begin{proof}[of Theorem~\ref{th:nphard-dfa}]
We give a reduction from 3SAT to the complement of the safety
verification problem.  Given a 3SAT formula with $n$ variables and \(m\) clauses, we
construct a deterministic leader \(\D\) and a deterministic
contributor \(\C\) such that the \( (\D,\C)\)-network writes a special
symbol $\#$ if{}f the formula is satisfiable.  
The leader uses the contributors to guess and store
an assignment, and then checks if each clause is satisfied.  

\smallskip %
\noindent %
{\it Gadget to guess and retrieve a bit.} %
The reduction uses the following protocol between the leader and the
contributors to guess a bit and maintain the guess consistently. 
To assign a value to {\tt bit}, the leader writes {\tt bit} to the global store.  
The contributors who read {\tt bit} from the store then write consecutively {\tt
propose-bit-is-\(i\)}, $i=0,1$, on the store.  The leader reads the store at
some (non-deterministic) point, and reads the last write by one of the
contributors proposing either $0$ or $1$.  If it reads {\tt propose-bit-is-0}
(the $1$ case is identical), it writes back that it commits to setting the bit
to $0$ (writing {\tt commit-bit-is-0}).  Contributors who read {\tt
commit-bit-is-0} move on to the next phase, where they deliver {\tt bit-is-0}
each time they are asked the value of {\tt bit}.  That is, they wait to read a {\tt
get-value-of-bit} message, and reply with {\tt bit-is-0}.

Similarly, if the leader commits to a $1$, contributors who read the message
come to the consensus that the bit is $1$.
Contributors who miss the commit message are stuck. 
This protocol ensures that the leader and contributors can reach consensus on the value
of a bit, and even though they are deterministic, the value of the bit is chosen non-deterministically,
based on when the leader reads a value ({\tt propose-bit-is-\(i\)}, $i=0,1$) from the store. 
Notice that an arbitrary number of contributors can participate and potentially overwrite
each other, but the bit is fixed to a chosen value.
\smallskip %
\noindent %
{\it Reduction from 3SAT.} %
The leader uses the above protocol to ``assign'' non-deterministically chosen
consensus values to each variable \(x_1, \ldots, x_n\).  Then, it checks
sequentially that each clause is satisfied by this assignment. 
To do this, it gets the literals from the contributors and checks if the clause is satisfied.  
To get the assigned value to a variable $x$, the leader writes $\mbox{\tt get-value-of-}x$
on the store. The contributors that are storing an assignment to $x$ (i.e.,
those who completed the consensus protocol for $x$) and who read this message, 
write the consensus value (using values $x\mbox{\tt -is-}0$ or $x\mbox{\tt -is-}1$). 
Even if
several contributors write, they write the same value.

Suppose the formula is satisfiable. Then, there is an execution of the protocol
where the contributors reach a consensus for each bit corresponding to a
satisfying assignment, and the leader succeeds in checking all clauses.
Then, the value \(\#\) gets written to the store and the \( (\D,\C)\)-network is accepting.  
On the other hand, if the formula is not satisfiable, then the 
leader never succeeds checking all clauses and \(\#\) never gets written.
Note that the size of \(\D\) and \(\C\) is \(O(m+n)\) and that they are deterministic.
\qed
\end{proof}

%
%
\section{coNP upper bound}\label{sec:npNotSoEasiness}

We first recall some basic structural property of FSAs. Define the
reachability relation of a FSA as follows: state \(s_1\) is reachable from
\(s_2\) if{}f there is a (possibly empty) path from \(s_2\) to \(s_1\). We
define a strongly connected component (\texttt{scc}) of a FSA as an equivalence
class for the mutual reachability relation.

\begin{lemma}
	For every FSA \(A = (\Sigma, Q, \delta, q_0, F) \) there exists a finite
	collection \(A_1, \ldots, A_d\) such that each \(A_i\) satisfies the
	following properties:
\begin{compactenum}
	\item each \(A_i\) results from removing states and transitions from \(A\); and
	\item \(\bigcup_{i=1}^{d} L(A_i) = L(A)\); and
	\item for each \(v_1,v_2\in L(A_i)\) there exists \(v\in L(A_i)\) such that
		\(v \succeq v_1\) and \(v \succeq v_2\).
\end{compactenum}
	\label{lem:decomp}
\end{lemma}
\begin{proof}
	A path in a FSA is said to be \emph{accepting} if its first and last state are
	initial and final respectively.  Let \(\pi\) be an accepting path in \(A\)
	such that no state repeats in \(\pi\). 
	Define \(A_{\pi}\) to be the FSA that consists exactly of 
	\begin{inparaenum}[\upshape(\itshape 1\upshape)]
		\item the states and transitions of \(\pi\); and
		\item the states and transitions of the \texttt{scc}s visited
			by \(\pi\).
	\end{inparaenum}

	Clearly, \(A_\pi\) results from removing states and transitions from \(A\).
	Define the set \(\{A_{1},\ldots, A_{d}\}\) such that each accepting path
	\(\pi\) with no repeating state induces exactly one automaton \(A_{\pi}\in
	\set{A_{1},\ldots, A_{d}}\). This set is finite since there are only finitely
	many states and transitions in \(A\).
	Furthermore, it is easily checked that \(\bigcup_{i=1}^{d} L(A_i)=L(A)\).

	We turn to point 3. Because no state is repeated, \(\pi\) fixes a total order
	on the \texttt{scc}s it visits (and are also included in \(A_{\pi}\)).  So any
	two accepting paths in \(A_{\pi}\) must visit the \texttt{scc}s in that order.
	Therefore, it suffices to show that given a \texttt{scc}, any two traces drawn
	upon that \texttt{scc} are covered by a third one. This is easily seen from
	the definition of subword ordering and the fact that in a \texttt{scc} all
	states are mutually reachable.\qed
\end{proof}

Next, let $g \in \G$ and define \(LL_g= L_g \cdot w_c(g)^*\).
Observe that \(LL_g \subseteq \G(r_c,\uw_c)^*\; \fw_c(g)\; w_c(g)^* \).
Hence, the alphabet of \(LL_g\) is given by \(\G(r_c,\uw_c)\cup\set{\fw_c(g),w_c(g)}\).

\begin{lemma}
	Let \(v \in \G(r_d, w_d, \fw_c, w_c)^*\),
	\(g\in \G\), \(G_1\subseteq \G\) such that \(g\notin G_1\).
	We have: if \(v \parallel L(\ES) \parallel (\shuffle_{g'\in G_1} LL_{g'}) \neq
	\emptyset\) and \(v \parallel L(\ES) \parallel LL_{g} \neq \emptyset\) then
	\(v\parallel L(\ES) \parallel (\shuffle_{g\in G_1\cup\set{g}} LL_g) \neq
	\emptyset\).
	\label{lem:composition1}
\end{lemma}
\begin{proof}
	Let \(x \in \shuffle_{g'\in G_1} LL_{g'}\) and \(y\in LL_{g}\).
	We prove: if \(v \parallel L(\ES) \parallel x \neq \emptyset\) and
	\(v \parallel L(\ES) \parallel y \neq \emptyset\) then 
	\(v\parallel L(\ES) \parallel (x \shuffle y) \neq \emptyset\).

	Let \(v=v_0\ldots v_{n-1}\), where \(v_i\in \G(r_d, w_d, \fw_c, w_c)\) for
	each \(i\), \(0\leq i< n\), and consider the asynchronous product \(v
	\parallel L(S^E)\). We define the LTS \(Q\) over alphabet \(\Sigma_E\) as follows. 
	The states of \(Q\) are the set \(\set{ (V_i,g) \mid 0\leq i\leq n,\,
	g\in\G\cup\set{\$} }\). The initial state is \( (V_0,\$)\).
	There is a transition \( (V_i,g) \by{a} (V_j,g')  \) between two states of \(Q\) if{}f
	one of the following condition holds:
	\begin{compactitem}
		\item \(a = r_d(g) = v_i\), \(j=i+1\), \(g'=g\);
		\item \(a = r_c(g)\), \(j=i\), \(g'=g\);
		\item \(a = w_d(g') = v_i\), \(j=i+1\);
		\item \(a = \fw_c(g') = v_i\), \(j=i+1\);
		\item \(a = w_c(g') = v_i\), \(j=i+1\);
		\item \(a = \uw_c\), \(j=i\), \(g'=\$\).
	\end{compactitem} 
        It is easy to see that \(L(Q)\) is the prefix closure of the language \(v \parallel L(S^{E})\).
	
        Since \(v \parallel L(\ES) \parallel x \neq
	\emptyset\) and \(v \parallel L(\ES) \parallel y \neq \emptyset\) by hypothesis, there exist two words
	\(\sigma_x\) and \(\sigma_y\) of \(Q\) such that \(\sigma_x\parallel x \neq
	\emptyset\) and \(\sigma_y\parallel y \neq \emptyset\).
	Let \(\sigma'_x\) be the word resulting from erasing all symbols
	in \(\G(\uw_c,r_c)\) from \(\sigma_x\), and define \(\sigma'_y\) similarly.
	It can be easily checked that \(\sigma'_x=\sigma'_y=v\).
	Therefore, there exist \(\alpha_i, \beta_i \in \G(\uw_c,r_c)^*\) for every \(i\), \(0\leq i\leq n\) such that
	\begin{align*}
		\sigma_x &= \alpha_0 v_0 \alpha_1 v_1 \ldots \alpha_{n-1} v_{n-1} \alpha_n\\	
		\sigma_y &= \beta_0 v_0 \beta_1 v_1 \ldots \beta_{n-1} v_{n-1} \beta_n
	\end{align*}
	The definition of \(Q\) further shows that 
	if \(v_{i-1}\) is a read or write to \(g\) then	\(\alpha_i = r_c(g)^{i_1}
	\uw_c^{i_2}  \) for some positive integer \(i_1,i_2\) and similarly \(\beta_i =
	r_c(g)^{j_1} \uw_c^{j_2}  \) for some positive integer \(j_1,j_2\).
	Moreover, if \(i_2+j_2>0\) then \(v_{i}\) must be a write. Define for each \(0\leq i \leq n\) the word
	\(\gamma_i= r_c(g)^{i_1 + j_1} \uw_c^{i_2+j_2}\), and let
	\(\sigma_{sh}=\gamma_0 v_0 \ldots \gamma_{n-1} v_{n-1} \gamma_n \).
	It is routine to check that \(\sigma_{sh}\in L(Q)\) holds, and so
	\(L(Q)\parallel \sigma_{sh}\neq \emptyset\). 
	Let \(\sigma'_{sh}\) be the result of erasing from \(\sigma_{sh}\) all symbols
	of \(\G(r_d,w_d)\) (which necessarily correspond to symbols of \(v\)). 
	Again we find that \( L(Q)\parallel \sigma'_{sh}\neq \emptyset\).  Finally,
	by the definition of \(\gamma_i\) and since \(g \notin G_1\) we get
	\(\sigma'_{sh}\in (x\shuffle y)\), and so \(L(Q)\parallel (x\shuffle
	y)\neq\emptyset\). By the definition of \(Q\) this implies \(v\parallel L(S^{E}) \parallel (x\shuffle
	y)\neq \emptyset\), and we are done.\qed
\end{proof}

\begin{proof}[of Thm.~\ref{thm:npnotsoeasy}]
\begin{compactenum}
	\item guess \(\tau=g_1\ldots g_d \in \Upsilon_{\#}\), in particular \(g_d=\#\);
	\item compute a FSA \(Q_{\tau}\) such that \(L(Q_{\tau})=P_{\tau}\cap L(\US)\); 
	\item compute the \texttt{scc}s of \(Q_{\tau}\); 
	\item guess an accepting path \(\pi_{\tau}\) of \(Q_{\tau}\) where no state repeats;
	\item compute the FSA \(q_{\tau}\) consisting exactly of the states and transitions of \(Q_{\tau}\) visited by \(\pi_{\tau}\) and the states and transitions of the \texttt{scc}s of \(Q_{\tau}\) visited by \(\pi_{\tau}\);
	\item check whether \(L(q_{\tau})\parallel L(S^{E}) \parallel LL_{g_i}\neq\emptyset\) for each \(i\), \(1\leq i\leq d\). 
\end{compactenum}

\noindent %
{\bf Non-deterministic polynomial time.} %
The two guesses are of polynomial size, the first in the size of \(\G\), the
second in the size of the FSA \(Q_{\tau}\) which can be computed in time polynomial
in the size of \(\G\) and \(\D\). Computing the \texttt{scc}s is done in time linear in the size of \(Q_{\tau}\) using Tarjan's algorithm \cite{CLR1990}.
Each of the \(d\) checks can be made in polynomial time.
To see this, we first notice that we can safely replace \(L(S^{E})\) by
\(L(S^{E})\cap P_{\tau}\) for which an FSM can be computed in polynomial time.
Therefore it is clear that each check can be made in polynomial time.

\noindent %
{\bf Correctness.} %
We want to show that 
\[ L(q_\tau) \parallel L(S^E) \parallel (\shuffle_{g\in \G} L(S_g)) \neq \emptyset 
\text{ if{}f }  L(q_\tau) \parallel L(S^E) \parallel LL_{g_i}\neq \emptyset  \text{ for each }i, 1\leq i\leq d\enspace .\]
For the only if direction it suffices to notice that by Lem.~\ref{lem:erasable} and
definition of \(P_{\tau}\) we can restrict the shuffle to those values occurring
in \(\tau\) only. Hence, the definition of \(L(S_g)\) and the shuffle operator
show the rest.

For the if direction, let \(v_1,\ldots,v_d\) be the \(d\) words from
\(L(q_{\tau})\) such that \( v_i \parallel L(S^{E}) \parallel LL_{g_i} \neq
\emptyset\) for each \(i\), \(1\leq i\leq d\).  By repeatedly applying
Lemma~\ref{lem:decomp} point 3 we find there exists \(v\in L(q_{\tau})\) such that \(v\succeq v_i\) for
every \(i\). Furthermore, Lemma~\ref{lem:coversuffices} shows that 
\(v \parallel L(S^{E}) \parallel LL_{g_i} \neq \emptyset\) for each \(i\).

Finally, we conclude from repeated applications of Lemma~\ref{lem:composition1}
that \( v\parallel L(S^{E})\parallel (\shuffle_{i=1}^d LL_{g_i})\neq \emptyset\),
hence that \(L(q_\tau) \parallel L(S^E)
\parallel (\shuffle_{g\in \G} L(S_g)) \neq \emptyset\) by definition of
\(L(S_g)= \pre(LL_g)\) and we are done.\qed
\end{proof}

%
%

\section{PSPACE lower bound}\label{sec:pspacehard}

\begin{proof}[of Theorem~\ref{th:pspace-hard}]
We give a reduction from the acceptance problem of a linear space-bounded
deterministic Turing machine to the complement of the safety verification
problem.  Fix a deterministic TM $M$ that on input of size
$n$ uses at most $n$ tape cells and accepts in exactly $2^{n}$ steps.  We are
given an input $x$ and want to check if $M$ accepts $x$.  An accepting run is a
sequence of TM configurations \(c_0 \rightarrow c_1 \rightarrow \ldots
\rightarrow c_{2^n}\), where $c_0$ is the initial configuration (the input $x$
is written on the tape, the head points to the leftmost cell, and the TM is in
its initial state), there is a transition of the TM from $c_i$ to $c_{i+1}$ for
$i=0,\ldots,2^n-1$, and $c_{2^n}$ is accepting.  Following the above
assumptions, configurations of \(M\) can be encoded by words of fixed length.

We define a \((\D,\C)\)-network that simulates \(M\) and such that
\(\D=\lts{P_D}\) and \(\C=\lts{P_C}\) where \(P_D\) and \(P_C\) are dPDMs.
The leader and the contributors co-operatively simulate computations of \(M\) using their stack, and
also use the stack to count up to $2^n$ steps. 
We start by describing the basic gadgets used in the simulation.

\smallskip %
\noindent %
{\it Counting to $2^n$ using $n$ stack symbols.} %
We show how a contributor can use its stack to count down from $2^n$.
Consider a stack alphabet of with $n+2$ symbols
$\set{\mathsf{I}_0,\ldots, \mathsf{I}_{n}} \cup \set{\$}$, where $\$$ is a special
bottom-of-stack marker.
Given a stack over this alphabet,
the contributor PDM, provided the top of the stack is \(\mathsf{I}_i\) for some \(0\leq i\leq n\), performs a \texttt{decrement} operation defined as follows:
\begin{compactenum}
\item While the top of the stack is $\mathsf{I}_i$ for some $i>0$, do \texttt{pop}($\mathsf{I}_i$) \texttt{;} \texttt{push}($\mathsf{I}_{i-1}$) \texttt{;} \texttt{push}($\mathsf{I}_{i-1}$)\texttt{;}
\item \texttt{pop}($\mathsf{I}_0$) and \texttt{return};
\end{compactenum}
Suppose initially the stack contains $\mathsf{I}_{n}\$$ (the bottom of the stack is to the right).
Then, we reach a stack with $\$$ on the top exactly after popping \(\mathsf{I}_0\) $2^n$ times, that is after performing \(2^n\) times the \texttt{decrement} operation.

\smallskip %
\noindent %
{\it Computing one step of the TM.} %
In the construction, we simulate one step of the \(M\)
by sending a configuration from the leader to contributors,
and then sending back the next configuration from contributors to the leader.

Assume the reverse of a configuration of the Turing machine is
stored as a word $w$ of length $n$ in the stack of the leader and the stack
of the contributors is empty. 
We want to ensure that at the end of the protocol, the stack of the leader
contains the reversal of a successor configuration
of \(M\), and the stack of the contributor is again empty.

We use the following protocol.
The leader and contributors use their finite set of control states to
count till $n$.
The leader pops one symbol of $w$ at a time from its stack, writes it
on to the global store, and waits for an acknowledgment from some 
contributor that the symbol has been received.
Conversely, the contributors read the letters of $w$ one symbol at a time from
the global store.
Moreover, using its finite state, the contributors compute the successor configuration
of the configuration that is received from the leader and store it on to their stacks. 
Additionally, a contributor sends an acknowledgment for the receipt of each symbol read
from the leader.

After $n$ steps of the leader and contributors, the stack of the leader is empty and the
stack of the contributor contains $w'$ where $w'$ can be reached from $w^R$ by
executing one step in \(M\).

Notice that at the end of this part of the protocol, in spite of asynchronous reads and writes,
the leader is certain that all $n$ symbols were received in order,
but not necessarily by the same contributor.

The second part of the protocol sends this configuration back from a contributor to the leader.
Again, the leader and the contributor use their finite state to count till $n$.
The contributor sends $n$ symbols one at a time to the leader, and waits for an acknowledgement
to check that the leader read the same symbol it transferred.
After $n$ steps, the leader's stack contains the reverse of $w'$ and the contributor's stack is again
empty.
Moreover, the contributor is certain that the entire configuration has been correctly received by the
leader.

Notice that even in the presence of non-atomic reads and writes, if the leader and {\em some} contributor 
successfully reach the end of the protocol, then the leader and that particular contributor has
faithfully simulated one step of the machine.
However, we cannot ensure that the same contributor participated in one whole round of the protocol,
always reading and writing the latest values and faithfully simulating one step of the Turing machine.
For example, it is possible that several contributors, that have simulated the Turing machine
for different number of steps, participate in the protocol. 
The simulation catches these discrepancies by counting, as described below.

\smallskip %
\noindent %
{\it The reduction.} %
Initially, all contributors push $\mathsf{I}_{n}\$$ onto their stacks.
The leader pushes $\mathsf{I}_{n}\$$ onto the stack, and additionally,
the reverse of the starting configuration of the Turing machine.

Then, the leader and contributors execute the protocol described above.
At the end of the first part of a round, the leader perform a \texttt{decrement}.
At the end of the second part of a round, the contributor perform a \texttt{decrement}.

The network accepts the computation (e.g., by outputting a special
symbol $\#$) if (1) both the leader and some contributor
count up to $2^n$, and (2) at that point, the Turing machine is in an accepting configuration.

Notice that if the leader interacts with the same contributor for $2^n$ rounds, then
both of them will simultaneously reduce the counter on the stack to $\$$ at the same
time, and thus, would have correctly simulated the Turing machine for $2^n$ steps.
So, if the stack encodes an accepting configuration, the Turing machine accepts.

Conversely, if the leader interacts with multiple contributors in different rounds,
then there will not be any contributor whose count reaches $2^n$ simultaneously with
the leader.
All such computations are not faithful simulations of the Turing machine and none
of them therefore lead to accepting the computation.

Finally, we note that in the above reduction, all processes are deterministic machines.
\qed
\end{proof}



%
%

\section{Language Theoretic Constructions}\label{sec:language}

We now complete the proof of Theorem~\ref{thm:pspaceeasy} by providing the
language-theoretic constructions.  We assume familiarity with basic formal
language theory \cite{sipser}.

A \emph{context-free grammar} (CFG) is a tuple $G=(\mathcal{X},\Sigma,\prod,X_0)$ where
$\mathcal{X}$ is a finite set of \emph{variables} containing the \emph{axiom} \(X_0\), $\Sigma$ is an
alphabet, $\prod \subseteq \mathcal{X}\times (\Sigma\cup \mathcal{X})^*$ is a
finite set of \emph{productions} (the production $(X,w)$ may also be noted
$X\rightarrow w$).
The size of a production $X\rightarrow w$ is $|w|+2$.
The size \(|G|\) of a CFG \(G\) is the sum of all the sizes of productions in $\prod$.
A CFG \(G=(\mathcal{X},\Sigma,\prod,X_0)\) is in Chomsky normal form (CNF) if{}f
$\prod \subseteq \bigl(\mathcal{X} \times (\Sigma \cup \mathcal{X}^2)\bigr) \cup \set{(X_0,\varepsilon)}$.
A CFG can be converted to CNF in time polynomial in its size. 

Given two strings $u,v \in (\Sigma \cup \mathcal{X})^*$ we define a \emph{step} relation
$u \Rightarrow v$ if there exists a production $(X, w)\in\prod$ and some words
$y,z \in (\Sigma \cup \mathcal{X})^*$ such that $u=yXz$ and $v=ywz$.
A step is further said to be \emph{leftmost} if \(y\in \Sigma^*\), that is the
production is applied on the leftmost variable of \(u\).
We use
$\Rightarrow^*$ to denote the reflexive transitive closure of $\Rightarrow$.
The \emph{language} of \(G\) is \(L(G)=\set{w\in\Sigma^*\mid X_0\Rightarrow^*
w}\) and we call any sequence of steps from \(X_0\) to \(w\in \Sigma^*\)
a \emph{derivation}. 
A derivation is \emph{leftmost} if it is a sequence of leftmost steps. 

Given a CFG with relation $\Rightarrow$ between strings,
for every \(k\geq 1\) we define the subrelation
\(\By{k}\) of \(\Rightarrow\) as follows: 
\(u \By{k} v\) if{}f \(u\Rightarrow v\) and both \(u\) and \(v\)
contain at most $k$ occurrences of variables.  We denote by
\(\Bystar{k} \) the reflexive transitive closure of
\(\By{k}\).  The \emph{\(k\)-index language} of \(G\) is
\(L^{(k)}(G)=\{w\in\Sigma^* \mid X_0 \Bystar{k} w\}\) and we
call the sequence of steps from \(X_0\) to \(w\in\Sigma^*\) a \emph{\(k\)-index
derivation}.  

The following properties holds: \(\By{k} \subseteq \By{k+1}\) for all
\(k\geq 1\); if \(B \Bystar{k-1} w\) then \( BC
\Bystar{k} w C\); moreover, if \( BC \Bystar{k} w\), then there exist \(w_1, w_2\) such that
\(w = w_1 w_2\) and either \((i)\) \(B \Bystar{k-1} w_1 \),
\(C \Bystar{k} w_2\), or \((ii)\) \(C
\Bystar{k-1} w_2\) and \(B \Bystar{k} w_1 \).

\smallskip %
\noindent %
{\it Asynchronous product of CFGs and FSAs.} %
Given a CFG $G$ and a FSA $A$, 
we now define a CFG $G \bowtie A$ such that $L(G\bowtie A) = L(G) \parallel L(A)$.
Without loss of generality, we assume the set of accepting states of
\(A\) is a singleton. We further show that the \(k\)-index language of
\(G \bowtie A\) is the asynchronous product of the \(k\)-index
language of \(G\) and the language of $A$.

\begin{definition}
  \label{def:bowtie}
  Given a CFG \(G=(\mathcal{X},\Sigma_g,\prod,X_0)\) in CNF and a FSA
  \(A=(\Sigma_a,Q,\delta,q_0,\set{q_f})\), we define 
  \(G\bowtie A\) as the CFG \(G^{\bowtie}=(\mathcal{X}^{\bowtie}, \Sigma_{\bowtie}, \prod^{\bowtie}, X_0^{\bowtie})\).
  First replace in \(G\) every production of the form
  \(X\rightarrow \sigma(\in\Sigma_g\cup\set{\varepsilon})\) by two productions \(X\rightarrow \sigma \bot\) and \(\bot \rightarrow \varepsilon\)
  where \(\bot\) is a variable not in \(\mathcal{X}\). This modified
  grammar is again referred to as \(G=(\mathcal{X},\Sigma_g,\prod,X_0)\).

  Then  define \(G^{\bowtie}=(\mathcal{X}^{\bowtie},\Sigma_{\bowtie},\prod^{\bowtie},X_0^{\bowtie})\) as follows:
  \begin{compactitem}
    \item \(\mathcal{X}^{\bowtie}= Q\times \mathcal{X}\times Q\); \(\Sigma_{\bowtie}=\Sigma_g\cup \Sigma_a\); \(X_0^{\bowtie}= \tuple{q_0,X_0,q_f}\);
    \item \(\prod^{\bowtie}\) contains no more than the following transitions:
      \begin{compactitem}
        \item if \(X\rightarrow \sigma \bot \in \prod\)  and \(\sigma\notin \Sigma_a\) then \( \tuple{q,X,q'} \rightarrow \sigma \tuple{q,\bot,q'}\in\prod^{\bowtie} \)
        \item if \(\sigma\notin \Sigma_g\) and \( (q,\sigma,q')\in \delta\) then \(\tuple{q,X,q''} \rightarrow \sigma \tuple{q',X,q''}\in\prod^{\bowtie} \)
        \item if \(X\rightarrow \sigma \bot \in \prod\), \((q,\sigma,q')\in \delta\) and \(\sigma\neq\varepsilon\) then \(\tuple{q,X,q''} \rightarrow \sigma \tuple{q',\bot,q''}\in\prod^{\bowtie}
	  \)
        \item if \(X\rightarrow Y Z \in \prod\) then \(\tuple{q_1,X,q_2} \rightarrow \tuple{q_1,Y,q'}\tuple{q',Z,q_2}\in\prod^{\bowtie} \)
        \item if \(q\in Q\) then \(\tuple{q,\bot,q}\rightarrow \varepsilon\in\prod^{\bowtie}\)
      \end{compactitem} 
  \end{compactitem} 
\end{definition}

\begin{proposition}
Let \(G\), \(A\) and \(G^{\bowtie}\) as in def.~\ref{def:bowtie}. We have
$L^{(k)}(G \bowtie A) = L^{(k)}(G) \parallel L(A)$ for every $k \geq 1$,
hence $L(G \bowtie A) = L(G) \parallel L(A)$. 
Moreover, \(G^{\bowtie}\) is computable in time polynomial in \(|G|+|A|\).
\label{prop:bowtie}
\end{proposition}
\begin{proof}
It suffices to show that for each \(q,q'\in Q\), \(X\in\mathcal{X}\) and \(k\geq 1\): \(\tuple{q,X,q'}\By{k}^* w\)
if{}f \(w\in w^a \parallel w^g\)  and \(X\By{k}^* w^g\) in
\(G\) and \(q\by{w^a} q'\) in \(A\).
From Def.~\ref{def:bowtie} it is clear that \(G^{\bowtie}\) is computable
in time polynomial in \(|G|\) and \(|A|\).
\qed
\end{proof}

\smallskip %
\noindent %
{\it Computing a FSA supporting \(L(G)\).} %
We first show that, given a CFG $G$, 
one can construct a FSA accepting a support of \(L(G)\) and whose
size is at most exponentially larger than the size of $G$.

\begin{theorem}{}\label{thm:min-cfl}
  Given a CFG $G=(\mathcal{X},\Sigma,\prod,X_0)$ in CNF with $n$ variables, we can compute 
a FSA $A$ with $O(2^{n \log(n)})$ states such that \(L(A)\) is a support of \(L(G)\). 
\end{theorem}

The proof of Theorem~\ref{thm:min-cfl} requires the following technical lemma.

\begin{lemma}
Let $G=(\mathcal{X},\Sigma,\prod,X_{\imath})$ be in CNF and \(D\colon X_0
\Rightarrow^* v\in \Sigma^*\) a leftmost derivation for some
\(X_0\in\mathcal{X}\). There exists $v' \preceq v$ and a leftmost \(n\)-index
derivation \(D'\colon X_0 \Bystar{n} v'\), where \(n\) is the number of
distinct variables appearing in \(D\). 
\label{lem:mininduction}
\end{lemma}
\begin{proof}
	By induction on the number \(m\) of sequences of steps of the
	form \(X \Rightarrow^* w X \alpha \Rightarrow^* w w' \alpha\) with
	\(w\neq\varepsilon\) occurring in \(D\).

	{\it Basis.} \(m=0\).
	The proof for this case is by induction on the number \(n\) of distinct
	variables appearing in \(D\).

	{\it Basis.} \(n=1\). Because \(G\) is in CNF and the assumption \(m=0\), \(D\) necessarily is such that \(X_0\Rightarrow v\in\Sigma\). Hence setting \(v'=v\) we find
	that \(X_0\Bystar{n}v'\) which concludes the case.

	{\it Step.} \(n>1\). Because \(G\) is in CNF and the assumption \(m=0\), it must
	be the case that \(D\) has the following form \(X_0\Rightarrow B C
	\Rightarrow^* w_1 C \Rightarrow^* w_1 w_2 = v\). Moreover \(m=0\) shows that
	\(X_0\) does not appear in the subsequence of steps \(D_1\colon
	B\Rightarrow^* w_1\) and \(D_2\colon C \Rightarrow^* w_2\).
	The number of distinct variables appearing in \(D_1\) and \(D_2\) being at
	most \(n-1\) we conclude, by induction hypothesis, that there exists leftmost
	\( (n-1)\)-index derivations \(D'_1\colon B\Bystar{n-1} w'_1\) and \(D'_2\colon C\Bystar{n-1}
	w'_2\) with \(w'_1 w'_2 \preceq w_1 w_2 = v\), hence that there exists a
	derivation \(D'\colon X_0 \Bystar{n} B C \Bystar{n} w'_1 C \Bystar{n} w'_1
	w'_2 =v'\) such that \(v'\preceq v\) and we are done.

	{\it Step.} \(m>0\). 
	Therefore in \(D\) there exists some variable \(X\) such that
	\(X \Rightarrow^* w X \alpha \Rightarrow^* w w' \alpha\) and
	\(w\neq\varepsilon\). Define the derivation \(D'\) given by \(D\) where the
	above subsequence of steps is replaced by \(X\Rightarrow^* w'\). Clearly we
	have that the word \(v'\) produced by \(D'\) is a subword of \(v\), the word
	produced by \(D\).  Moreover the above transformation on \(D\) allows to use
	the induction hypothesis on \(D'\), hence we find that there exists there
	exists a leftmost \(n\)-index derivation \(D''\colon X_0 \Bystar{n} v''\) and
	\(v''\preceq v'\) and we are done since \(v''\preceq v'\preceq v\).
	\qed
\end{proof}

\begin{proof}[of Theorem~\ref{thm:min-cfl}]
From Lem.~\ref{lem:mininduction} it is easy to see 
that the words given by the leftmost \(n\)-index derivations is a support of \(L(G)\).
Recall that \(G\) is in CNF. Next we define a FSA \(A=(\Sigma,Q,\delta,q_0, F)\) such that 
\begin{inparaenum}[\itshape i\upshape)]
	\item \(Q=\set{ w \in \mathcal{X}^j \mid 0\leq j\leq n}\);\linebreak
	\item \(\delta=\set{(\alpha\, w,\gamma, \beta\, w)\mid (\alpha,\gamma\beta)\in\prod\land \gamma\in\Sigma\cup\set{\varepsilon}}\);
	\item \(q_0=X_0\);
	\item \(F=\set{\varepsilon}\).
\end{inparaenum}
It is easy to see that \(A\) simulates all the leftmost sequence of steps of
index at most \(n\) and accepts only when those corresponds to
\(n\)-index leftmost derivations, hence that \(L(A)\) is a support of \(L(G)\).
Also since \(n=|\mathcal{X}|\), \(Q\) has \(O(n^n)\) states, or equivalently \(O(2^{n\log(n)})\).
\qed
\end{proof}

From the construction, it is clear that there is a polynomial-space bounded algorithm
(in $|G|$) that can implement the transition relation of $A$, that is, given an
encoding of a state of $A$, produce iteratively the successors of the state.

\smallskip %
\noindent %
{\it Covering Context-free Languages by Bounded-Index Languages.} %
Our second construction shows that, given a CFG $G$, we can construct
an $O(|G|)$-index language that is a cover of $L(G)$.

Given a CFG $G=(\mathcal{X},\Sigma,\prod,X_0)$ 
and $X, Y \in \mathcal{X}$, we say that
$Y$ \emph{depends on} $X$ if $G$ has a production $X \rightarrow \alpha Y \beta$ 
for some $\alpha, \beta \in (\mathcal{X} \cup \Sigma)^*$, or if there is a 
variable $Z$ such that $Y$ depends on $Z$ and $Z$ depends on $X$. 
A {\em strongly connected component} (SCCs) of $G$ if a maximal subset 
of mutually dependent variables.

\begin{theorem}{}\label{thm:bamboo}
Let $G=(\mathcal{X},\Sigma,\prod,X_0)$ be a CFG in
CNF with $n$ variables and $k$ SCCs. 
Then \(L^{(n+2k)}(G)\) is a cover of \(L(G)\). 
\end{theorem}

The proof of Theorem~\ref{thm:bamboo} uses the following technical lemma which
follows from the fact that the commutative images of $L(G)$ and $L^{(n+1)}(G)$
coincide \cite{egkl11-ipl}.

\begin{lemma}
\label{lem:cover}
Let $G=(\mathcal{X},\Sigma,\prod,X_0)$ be a CFG in CNF with $n$ variables.
For every $a \in \Sigma$, if $X_0 \Rightarrow^* w\, a\, v$ for some $w, v \in \Sigma^*$, then 
$X_0 \By{n+1}^* w' \, a \, v'$ for some $w', v' \in \Sigma^*$.
\end{lemma}
\begin{proof}[of Theorem~\ref{thm:bamboo}]
Let $w \in L(G)$. 
If $w=\varepsilon$, then since $G$ is in CNF, we have $X_0\Rightarrow\varepsilon$.
So, $w$ is also in $L^{(n+2k)}(G)$.
Hence, assume $w\neq\varepsilon$.
We prove by induction on $k$ that $w$ is a subword of some $w' \in L^{(n+2k)}(G)$.

{\em Basis.} $k=1$. 
We proceed by induction on $|w|$. Let $w=a$ for some $a \in \Sigma$.
We conclude from \(a\in L(G)\) and \(G\) is in CNF that \(X_0\Rightarrow a\), hence that \(X_0 \Bystar{1} a\) and finally that \(a\in L^{(n+2k)}(G)\) and we are done.
If $w=av$ for some $v \neq \varepsilon$ 
then 
$X_0 \Rightarrow X_1X_2 \Rightarrow^* a\,v$ for some $X_1, X_2 \in \mathcal{X}$. 
By Lemma~\ref{lem:cover},
$X_1 \Bystar{n+1} v_1$ for some $v_1 \succeq a$, and by induction hypothesis 
$X_2 \Bystar{n+2} v_2$ for some $v_2 \succeq v$. So we get
$X_0 \By{2} X_1X_2 \Bystar{n+2} v_1X_2 \Bystar{n+2} v_1 v_2$
and taking $w' = v_1v_2$ we have $X_0 \Bystar{n+2} w' \succeq w$.

{\it Step.} $k > 1$. Let $\mathcal{Y} \subset \mathcal{X}$ be the set of 
variables of a bottom strongly connected component of $G$, and let $\prod_\mathcal{Y} \subset \prod$ be the productions of $G$ with a variable of $\mathcal{Y}$ on the left side. For every $Y \in \mathcal{Y}$, let $G_Y = (\mathcal{Y},\Sigma,\prod_\mathcal{Y}, Y)$; further, let $G'=(\mathcal{X} \setminus \mathcal{Y},\Sigma \cup \mathcal{Y},\prod \setminus \prod_\mathcal{Y},X_0)$. Since $w \in L(G)$, there exist derivations
$X_0 \Rightarrow^* w_1 Y_1 w_2 \ldots w_r Y_r w_{r+1}$ in $L(G')$ and $Y_i \Rightarrow^* v_i$ 
in $L(G_{Y_i})$ for every $i=1, \ldots, r$ such that $w_1v_1\ldots w_rv_rw_{r+1} = w$. 
Since $G'$ has $(k-1)$  SCCs, by induction hypothesis there is
$X_0 \Bystar{i_1} w_1' Y'_1 w_2' \ldots w_t' Y'_t w_{t+1}'$ in $L(G')$, where $i_1= n - |\mathcal{Y}| + 2(k - 1)$
and such that \(w_1' Y'_1 w_2' \ldots w_t' Y'_t w_{t+1}'\succeq w_1 Y_1 w_2 \ldots w_r Y_r w_{r+1}\).
In particular we have \(Y'_1 \dots Y'_t \succeq Y_1 \dots Y_r\) which implies that there
exists a monotonic injection \(h \colon \{1,\dots,r\} \rightarrow \{1,\dots,t\}\)
such that \(Y'_{h(i)}=Y_{i}\) for all \(i\in\set{1,\dots,r}\).
Since every $G_Y$ has one strongly connected component, for every $j=1, \ldots, r$ there is $v_j' \succeq v_j$ such that $Y'_{h(j)}=Y_j \By{i_2}^* v_j'$, where $i_2=|\mathcal{Y}|+2$. 
On the other hand, for every \(\ell\) not in the image of \(h\) there also is some word \(v'_{\ell}\) 
such that $Y'_{\ell} \By{i_2}^* v'_{\ell}$, where $i_2=|\mathcal{Y}|+2$. 

So we have 
\[
\begin{array}{lcl}
X_0 & \Bystar{i_1}     &w_1' Y'_1 w_2' Y'_2  \ldots w_t' Y'_t w_{t+1}'\\
    & \Bystar{i_1+i_2} &w_1' v_1' w_2' Y_2 \ldots w_t' Y_t w_{t+1}' \\
    & \Bystar{i_1+i_2} &w_1' v_1' w_2' v_2'\ldots w_t' Y_t w_{t+1}' \\
    & 		       \cdots \\
    & \Bystar{i_1+i_2} &w_1' v_1' w_2' v_2'\ldots w_t' v_t' w_{t+1}'
\end{array}
\]
Let $w'= w_1' v_1' \ldots w_t' v_t' w_{t+1}'$. Since $i_1+i_2= n+2k$, we get
$X_0 \Bystar{n+2k} w' \succeq w$, and we are done. 
\qed
\end{proof}

\smallskip %
\noindent %
{\it Checking Emptiness of $k$-index languages.} %
Finally, we show that \(L^{(k)}(G)\neq\emptyset\) is decidable in
\(\mathrm{NSPACE}(k \log(|G|))\).
In contrast, non-emptiness checking for context-free languages is
P-complete as shown by Jones \cite{Jonespcomplete}.

\begin{theorem}{}\label{thm:emptiness}
  Given a CFG $G$ in CNF and \(k\geq 1\), it is in \(\mathrm{NSPACE}(k \log(|G|))\) to decide whether $L^{(k)}(G) \neq \emptyset$.
\end{theorem}
\begin{proof}
We give a non-deterministic space algorithm. 
The algorithm, called $\mathit{query}$, takes two parameters, 
a variable $X\in\mathcal{X}$ and a number $\ell\geq 1$,
and guesses an $\ell$-index derivation of some word starting from $X$.
To do so, the algorithm guesses a production $(X,w)\in \prod$ with head \(X\).
If $X\rightarrow \sigma$ is chosen, for
$\sigma\in\Sigma\cup\set{\varepsilon}$, it returns true.
If $X\rightarrow BC$ is chosen, the algorithm non-deterministically
looks (using a recursive call)
\begin{inparaenum}[(i)]
\item\label{item:i} for an $(\ell-1)$-index derivation from $B$ and an $\ell$-index derivation from $C$, or
\item\label{item:ii} for an $(\ell-1)$-index derivation from $C$ and an $\ell$-index derivation from $B$.
\end{inparaenum}
When \(\ell=0\) or a recursive call returns false, then \(\mathit{query}\) returns false. 

We show the following invariant: \(\mathit{query}(\ell, X)\) has an
execution returning true if{}f \( X\Bystar{\ell} w\) for some \(w\in
\Sigma^*\). It follows that \(\mathit{query}(k,X_0)\) returns true
if{}f \(L^{(k)}(G)\neq\emptyset\).
The right-to-left direction is proved on the number \(m\) of steps in a bounded-index derivation.
If \(m=1\) then we have \( X\By{\ell} w\) with \(\ell=1\) and
\(\mathit{query}(\ell,X)\) returns true by picking \((X,w)\in \prod\). If \(m>1\) then the sequence of steps is
as follows \( X\By{\ell} BC \Bystar{\ell}
w\) where \(\ell\geq 2\). From there either \(B \Bystar{\ell-1} w_1
\), \(C \Bystar{\ell} w_2\), or \(C \Bystar{\ell-1} w_2\) and \(B
\Bystar{\ell} w_1 \) where \(w=w_1w_2\) holds.
Let us assume the latter holds (the other case is treated similarly).
Then we have \(C \Bystar{\ell-1} w_2\) and
\(B \Bystar{\ell} w_1\) and both sequence have no more
than \(m-1\) steps. Therefore the induction hypothesis shows that
\(\mathit{query}(\ell-1,C)\) and \(\mathit{query}(\ell,B)\) return
true, and so does \(\mathit{query}(\ell,X)\) by picking \((X,BC)\in \prod\).

The left-to-right direction is proved by induction on the number \(m\)
of times productions are picked in an execution of $\mathit{query}$ that returns true.
If \(m=1\) then \(\ell\geq 1\) and a production \((X,\sigma)\in\prod\)
where \(\sigma\in \Sigma\cup\set{\varepsilon}\) must have
been picked, hence the derivation \(X \By{\ell} \sigma\).
If \(m>1\), \(\mathit{query}\) recursively called itself after picking
\( (X,BC)\in \prod \).
Let us assume case (\ref{item:i}) was executed ((\ref{item:ii}) is treated similarly).
Following the assumption both calls \(\mathit{query}(\ell-1,B)\) and
\(\mathit{query}(\ell,C)\) return true (hence \(\ell\geq 2\)) and are such that 
productions are picked at most \(m-1\) times. 
Next, the induction hypothesis shows that \(B \Bystar{\ell-1} w_1\) and
\(C \By{\ell} w_2\) for some \(w_1\) and \(w_2\). Finally, \(X\rightarrow BC \) of \(\prod\) and \(\ell\geq 2\) shows that
\(X \By{\ell} BC \Bystar{\ell} w_1
C \Bystar{\ell} w_1 w_2\) and we are done.

It remains to show that \(\mathit{query}(k,X_0)\) runs in \(\mathrm{NSPACE}(k \log(|G|))\).  
Observe that for each non-deterministic choice (\ref{item:i}) or (\ref{item:ii}), there is one
recursive call $\mathit{query}(B,\ell-1)$ or $\mathit{query}(C,\ell-1)$.
The other call (e.g., to $\mathit{query}(C,\ell)$ in case (\ref{item:i})) 
is tail-recursive and can be replaced by a loop.
Since the index that is passed to that recursive
call decreases by \(1\) and \(\mathit{query}\) returns false when the
index is \(0\), we find that along every execution at most \(k\) stack frames are needed and each frame
keeps track of a grammar variable which can be encoded with \(\log(|G|)\) bits.
Hence we find that \(L^{(k)}(G)\neq \emptyset\) can be decided 
in \(\mathrm{NSPACE}(k \log(|G|))\).
\qed
\end{proof}
\begin{comment} 
\noindent
\begin{algorithm}
\NoCaptionOfAlgo
	\Begin{
	\lIf{\(\ell < 1\)}{\Return false}\;
	Let \(p\in\prod\) such that $\mathit{head}(p)=X$\;\nllabel{ln:pick}
\Switch{p}{
	\uCase{\(X\rightarrow \sigma\)}{\nllabel{ln:case0} 
		skip;
	} 
	\Case{\(X \rightarrow BC\)}{\nllabel{ln:case1}
		\eIf( \tcc*[f]{non-det.\ choice}){\(\ast\)}{
		\(\mathit{query}(\ell-1,B)\)\tcc*{case ({\it i})}
		\(\mathit{query}(\ell,C)\)\;
		}( \tcc*[f]{case ({\it ii})}){
		\(\mathit{query}(\ell-1,C)\)\;
		\(\mathit{query}(\ell,B)\)\;
		}
	}
	}
	\Return true\;\nllabel{ln:return}
}
\caption{\(\mathit{query}(\ell,X)\)\label{alg:nspace}}
\end{algorithm}%
\end{comment} 
\vspace{-.1cm}


\end{document}